\documentclass[journal]{IEEEtran}
\usepackage{hhline}

 \usepackage{cite}
\usepackage[cmex10]{amsmath}
\usepackage[caption=false,font=footnotesize]{subfig}
\usepackage{url}

\usepackage{amssymb, amsthm}

\usepackage{pgfplots}
\pgfplotsset{compat=1.10} 
\usepgfplotslibrary{statistics}

\pgfplotsset{mystyle/.style={%
    width=4.8cm, 
    height=4.8cm,
    xlabel near ticks,
    ylabel near ticks,
    tick label style={font=\footnotesize},
    label style={font=\footnotesize},
	yticklabel style={%
		/pgf/number format/precision=3,
		/pgf/number format/fixed},
	every axis y label/.style={%
		at={(rel axis cs:-0.25,0.5)},			
		rotate=90,								
		font=\footnotesize},					
    legend style={font=\footnotesize,
                  draw=none,					
                  overlay,at={(-0.3567,-.45)},	
                  anchor=north west},			
    legend columns=-1,							
}}

\pgfplotsset{Ell3style/.style={%
    width=8cm, 
    height=4cm,
    xlabel near ticks,
    ylabel near ticks,
    tick label style={font=\footnotesize},
    label style={font=\footnotesize},
    log ticks with fixed point,
    legend style={font=\footnotesize,
                  draw=none,
                  overlay,at={(-0.2,-.4)},
                  anchor=north west},
    legend columns=-1,
}}

\pgfplotsset{myboxplotstyle/.style={%
	cycle list={%
		{blue,mark=+},
		{blue,mark=+},
		{blue,mark=+},
		{blue,mark=+},
		{blue,mark=+},
		{blue,mark=+},
		{blue,mark=+},
		{blue,mark=+},			
		{blue,mark=+},				
		},
	boxplot/draw direction=y,
	xtick={1,2,3,4,5,6,7,8,9},
	xticklabels={LDA, Ell1-QDA,Ell2-QDA,LW-QDA, Gau-QDA,  Ell1-LDA, Ell2-LDA, LW-LDA, Gau-LDA},
	x tick label style={font=\footnotesize, rotate=-35},
	y tick label style={font=\footnotesize},
	width=9.5cm,
	height=5.5cm,
	}
}

\pgfplotsset{mystyleportfolio/.style={%
    width=4.8cm, 
    height=4.8cm,
    xlabel near ticks,
    ylabel near ticks,
    tick label style={font=\footnotesize},
    label style={font=\footnotesize},
	yticklabel style={%
		/pgf/number format/precision=3,
		/pgf/number format/fixed},
	every axis y label/.style={%
		at={(rel axis cs:-0.27,0.5)},			
		rotate=90,								
		font=\footnotesize},					
    legend style={font=\footnotesize,
                  draw=none,					
                  overlay,at={(-0.3567,-.5)},	
                  anchor=north west},			
    legend columns=-1,							
}}
\newcommand{\PDHm}[1]{\mathbb{S}_{++}^{#1\times#1}}
\newcommand{\pdim}{p}
\newcommand{\ndim}{n}
 
\newcommand{\M}{\bom \Sigma} 

\newcommand{\E}{\mathbb{E}}                  
\newcommand{\bo}[1]{\mathbf{#1}}              
\newcommand{\bom}[1]{\boldsymbol{#1}}  
\newcommand{\be}{\beta}
\newcommand{\iidsim}{\overset{iid}{\sim}}

\newcommand{\al}{\alpha}  
\newcommand{\gam}{\gamma}  
\newcommand{\beq}{\begin{equation}}
\newcommand{\eeq}{\end{equation}}
\newcommand{\bmat}{\begin{pmatrix}}
\newcommand{\emat}{\end{pmatrix}}
\newcommand{\etatwo}{\vartheta}

\newcommand{\bmu}{\bom \mu}

\renewcommand{\S}{{\bf S}}

\newcommand{\I}{\bo I} 

\renewcommand{\dim}{\pdim}

\newcommand{\R}{\mathbb{R}}

\newcommand{\X}{{\bf X}}

\newcommand{\Z}{{\bf Z}}

\newcommand{\A}{{\bf A}}
\newcommand{\B}{{\bf B}}

\newcommand{\commat}{\bo K_\pdim}  
\newcommand{\ve}{\mathrm{vec}} 

\renewcommand{\H}{{\bf H}}
\newcommand{\x}{{\bf x}}
\renewcommand{\u}{{\bf u}}

\renewcommand{\a}{{\bf a}}
\newcommand{\z}{{\bf z}} 

\renewcommand{\v}{{\bf v}}
\newcommand{\e}{{\bf e}}
\newcommand{\w}{{\bf w}}

\newcommand{\Fr}{\mathrm{F}}
\newcommand{\MSE}{\mathrm{MSE}}

\newcommand{\ka}{\kappa}

\newcommand{\kdim}{K}

\newcommand{\tr}{\mathrm{tr}}
\DeclareMathOperator{\Tr}{tr}

\DeclareMathOperator{\cov}{cov} 
\DeclareMathOperator{\var}{var} 

\newcommand{\Mor}{\bo S}

\newcommand\veco[1]{\mathrm{vec}\left({#1}\right)}
\newcommand\numberthis{\addtocounter{equation}{1}\tag{\theequation}}

\newcounter{ctheorem}
\newtheorem{theorem}[ctheorem]{Theorem}

\newcounter{clemma}
\newtheorem{lemma}[clemma]{Lemma}

\begin{document}
%
%
\title{Optimal shrinkage covariance matrix estimation under random sampling from elliptical distributions}

%
%
%

\author{Esa~Ollila,~\IEEEmembership{Member,~IEEE,}~and~Elias~Raninen,~\IEEEmembership{Student Member,~IEEE}
\thanks{E. Ollila and E. Raninen are with the Department of Signal
	Processing and Acoustics, Aalto University, Espoo, P.O. Box 13000,
	FIN-00076  Aalto; e-mail: \{esa.ollila,elias.raninen\}@aalto.fi 
(see http://users.spa.aalto.fi/esollila/), The research was supported by the Academy of Finland grant no. 298118 which
 is gratefully acknowledged}
\thanks{Manuscript received XX/20XX}}

%
%

\markboth{Submitted to IEEE TRANSACTIONS ON SIGNAL PROCESSING}{E. Ollila and E. Raninen}
%



\maketitle

\begin{abstract}
	This paper considers the problem of estimating a high-dimensional (HD)
	covariance matrix when the sample size is smaller, or not much larger,
	than the dimensionality of the data, which could potentially be very
	large. We develop a regularized sample covariance matrix (RSCM)
	estimator which can be applied in commonly occurring sparse data
	problems. The proposed RSCM estimator is based on estimators of the
	unknown optimal (oracle) shrinkage parameters that yield the minimum
	mean squared error (MMSE) between the RSCM and the true covariance
	matrix when the data is sampled from an unspecified elliptically
	symmetric distribution.  We propose two variants of the RSCM estimator
	which differ in the approach in which they estimate the underlying
	sphericity parameter involved in the theoretical optimal shrinkage
	parameter.  The performance of the  proposed RSCM estimators are
	evaluated with numerical simulation studies. In particular when the
	sample sizes are low, the proposed RSCM estimators often show a
	significant improvement over the conventional RSCM estimator by Ledoit
	and Wolf (2004). We further evaluate the performance of the proposed
	estimators in classification and portfolio optimization problems with
	real data wherein the proposed methods are able to outperform the
	benchmark methods. 
\end{abstract}


%
\IEEEpeerreviewmaketitle

%
%
%
%

\section{Introduction}\label{sec:intro}
\IEEEPARstart{E}{stimating} high-dimensional covariance matrices where the
sample size $\ndim$ is smaller, or not much larger than the dimension
$\pdim$ of the samples, is a problem that has attracted significant research
interest in the recent years~\cite{ledoit2002some,ledoit_wolf:2004,abramovich2007diagonally,ollila2014regularized,zhang2016automatic,sun2014regularized,chen2010shrinkage,couillet2014large,pascal2014generalized}.
This is due to the fact that high-dimensional data analysis problems have
become increasingly common in a wide spectrum of fields, such as in
finance~\cite{ledoit_wolf:2004},  bioinformatics, and
classification~\cite{friedman1989regularized}.

We consider the problem of estimating the high-dimensional covariance matrix
based on a sample  $\x_1,\ldots, \x_\ndim$ of independent and identically
distributed (i.i.d.) random vectors. The observations are assumed to be
generated from an unspecified $\pdim$-variate distribution $\x \sim F$ with
a mean vector $\bmu=\E[\x]$ and a $\pdim \times \pdim$ positive definite
covariance matrix
\[
	\M = \E \big[(\x - \bmu)(\x -\bmu)^\top \big] \in \PDHm{\pdim}.
\]
The most commonly used estimators of the unknown parameters $(\bmu, \M)
\in \R^\pdim \times \PDHm{\pdim}$ are the \emph{sample mean vector} and the
\emph{sample covariance matrix} (SCM),
\begin{align*}
	\bar \x &= \frac{1}{\ndim}\sum_{i=1}^\ndim \x_i , \\
	\S &= \frac{1}{\ndim-1}\sum_{i=1}^\ndim
	(\x_i- \bar \x) (\x_i - \bar \x)^\top.
\end{align*}
They have desirable properties, such as being the sufficient statistics and
unbiased. However, the SCM does not perform well in high-dimensional
problems for several reasons. Foremost, significant estimation errors result
from having an insufficient number of samples. Moreover, if $\pdim > \ndim$,
the SCM is always singular, i.e., it is not invertible even though the true
covariance matrix is known to be positive definite and hence non-singular.
In these situations, a frequently used approach for improving the estimation
accuracy is to use shrinkage regularization.

One of the most commonly used estimators in low sample support problems,
where $\pdim$ is large compared to the sample size $\ndim$, is the
\emph{regularized} SCM (RSCM) of the form
\begin{equation} \label{eq:regSCM}
	\S_{\al, \be} = \be \S + \al \I,
\end{equation}
where $\al, \be > 0$ denote the \emph{shrinkage  parameters} or
\emph{regularization parameters}. In signal processing, an estimator of the
form \eqref{eq:regSCM} was proposed in
\cite{carlson1988covariance,abramovich1981controlled} and is often referred
to as the diagonal loading estimator. Another line of research has been to
consider robust regularized covariance matrix estimators, e.g.,
\cite{abramovich2007diagonally,ledoit_wolf:2004,ollila2014regularized,zhang2016automatic,sun2014regularized,chen2010shrinkage,couillet2014large},
and~\cite{pascal2014generalized}.
In this paper, the focus is on determining the optimal (in MSE sense)
shrinkage parameters for the RSCM. 

We define the optimal RSCM estimator as the one that is based on the
\emph{oracle shrinkage parameters} minimizing the mean squared error (MSE),
that is,
\begin{equation} \label{eq:RegSCM_oracle}
    (\al_o,\be_o) 
    = \underset{\al,\be >0} {\arg \min}~\E \Big[ \big\|   \S_{\al,\be}  - \M \big \|^2_\Fr \Big],
\end{equation}
where $\| \cdot \|_\Fr$ denotes the Frobenius matrix norm,  i.e.,
$\|\A\|_\Fr^2 = \tr(\A^\top \A)=\tr(\A \A^\top)$ for any matrix $\A$. We use
the prefix \emph{oracle} for the shrinkage parameters $(\al_o,\be_o)$ as
they depend on the true unknown covariance matrix $\M$. Although, the oracle
shrinkage parameters cannot be used in practice, they have the theoretical
significance for being a benchmark for best possible performance w.r.t.\ the
MSE metric. 

The widely popular Ledoit-Wolf (LW-)RSCM~\cite{ledoit_wolf:2004} estimator
is based on the consistent estimators $(\hat \al_o^{\textup{LW}}, \hat
\be_o^{\textup{LW}})$ of the oracle parameters $(\al_o,\be_o)$ under the
random matrix theory (RMT) regime, i.e.,  as $\ndim,\pdim \to \infty$, we
have
	\begin{equation}\label{eq:RMTregime}
	c = \pdim/\ndim \to c_0, \tag{R1}
\end{equation}
It is, however, possible to improve upon the LW-estimator and obtain a more
accurate finite sample estimation performance by assuming that the
observations are generated from a specific $\pdim$-variate distribution,
e.g., the multivariate normal (MVN) distribution.  For example, in
\cite[Theorem~1]{chen2010shrinkage}, the authors derived an optimal
shrinkage parameter assuming that the samples have a Gaussian distribution
with a known location ($\bom \mu$).  Such a strict assumption on
distribution of the data however implies performance loss if the assumption
does not hold.  Another somewhat related approach has been taken for example
in \cite[Proposition~3]{auguin2016robust}, where the authors considered
robust $M$-estimators and looked for an asymptotically optimal shrinkage
parameter in the RMT regime which minimizes the squared Frobenius distance
between normalized regularized $M$-estimators of scatter matrix and a
normalized covariance matrix.

In this paper, we instead assume that the observations are from an
unspecified elliptically symmetric (ES) distribution and derive estimators
of the optimal oracle shrinkage parameters $(\al_o,\be_o)$ that are able to
perform reliably under the RMT regime. ES distributions is a large class of
distributions comprising, e.g., the MVN distribution, generalized Gaussian,
and all compound Gaussian distributions as special cases, see
e.g.,~\cite{muirhead:1982,fang1990symmetric}, and~\cite{ollila2012complex}.

The rest of this paper is organized as follows. In Section~\ref{sec:optim1}
and Section~\ref{sec:optim2}, we derive the optimal shrinkage parameters
$(\al_o,\be_o)$ under the general assumption of sampling from any
$\pdim$-variate distribution and an elliptical distribution with finite
fourth order moments, respectively. In Section~\ref{sec:estim}, we develop 
estimators of $(\al_o,\be_o)$ under the RMT regime and when sampling from an
unspecified elliptically symmetric distribution.
In Section~\ref{sec:simul}, we conduct several simulation studies and
compare the proposed estimators with the popular LW-RSCM estimator. In
Section~\ref{sec:real_data}, we illustrate the performance of the proposed
estimators in two applications. First, the proposed methods are used in a
regularized discriminant analysis (RDA) framework, where they are applied to
the classification of phoneme data. Then the methods are used in a portfolio
optimization problem using the Global Minimum Variance Portfolio (GMVP)
framework, where we use real data of historical (daily) stock returns from
the Hong Kong's Hang Seng Index (HSI) and  Standard and Poor's 500 (SP500)
index.  In both applications the proposed methods are shown to perform
better than the benchmark methods. Finally, Section~\ref{sec:conclusions}
concludes.

\emph{Notation:} We denote the open cone of $p \times p$ positive definite 
symmetric matrices by $\PDHm{\pdim}$. The \emph{vectorization} of an $n
\times p$ matrix $\A = (\a_1, \ldots \a_p)$ is denoted by $\ve(\A) =
{(\a_1^\top, \ldots, \a_p^\top)}^\top$. 
The \emph{matrix trace} of a square matrix $\A$ is denoted by $\tr(\A)$. The
\emph{Kronecker product} $\A \otimes \B$ of any matrices $\A$ and $\B$ is a
block matrix with its $ij$th block being equal to $a_{ij}\B$.  Kronecker
product has the useful property: $(\A \otimes \B)(\mathbf{C} \otimes
\mathbf{D}) = (\A\mathbf{C} \otimes \B\mathbf{D})$ for the matrices $\A$,
$\B$, $\mathbf{C}$, and $\mathbf{D}$ of appropriate dimensions. We denote
the \emph{identity matrix} of proper dimension by $\I$ and the
\emph{centering matrix} of proper dimension by $\H = \I -
\boldsymbol{1}\boldsymbol{1}^\top/n$, where $\boldsymbol{1}$ is a vector of
ones. The \emph{canonical basis vector}, which has its $i$th element equal
to $1$ and all other elements zero is denoted by $\e_i$. The
\emph{commutation matrix} $\commat$ is a $\dim^2 \times \dim^2$ block matrix
with its $ij$th block equal to a $\dim \times \dim$ matrix that has a $1$ at
element $ji$ and zeros elsewhere, i.e., $\commat = \sum_{i,j} \e_i \e_j^\top
\otimes \e_j \e_i^\top$. It also has the following important
properties~\cite{magnus_neudecker:1999}: $\commat \ve(\A)= \ve(\A^\top)$ and
$\commat (\A \otimes \B) \commat = (\B \otimes \A)$ for any $\pdim \times
\pdim$ matrices $\A$ and $\B$. In our developments, we will also use the
following identities: $\tr(\A \otimes \B) = \tr(\A) \tr(\B)$, $\tr( \ve(\A)
\ve(\B)^\top ) = \tr(\A^\top \B) = \ve(\B)^\top \ve(\A)$ for any square
matrices $\A$ and $\B$ of same order. Notation "$=_d$" reads "has the same
distribution as". 

\section{Optimal oracle shrinkage parameters}\label{sec:optim1}
In this section, we derive the oracle shrinkage parameters $(\alpha_o,
\beta_o)$ for any $p$-variate distribution. First, we define the 
\emph{scale} and  \emph{sphericity} parameters of $\M \in \PDHm{\pdim}$ as 
\begin{equation} \label{eq:scale}
	\eta = \frac{\tr(\M)}{\pdim} \quad \text{and} \quad \gamma =
	\frac{\pdim \tr(\M^2)}{\tr(\M)^2}.
\end{equation}
Note that  $\eta$ equals the mean of the eigenvalues  of $\M$ whereas $\gam$
is equal to the ratio of the mean of the squared eigenvalues relative to
mean of eigenvalues squared. The sphericity $\gamma$
\cite{ledoit2002some,srivastava2005some} measures how close the covariance matrix is to a
scaled identity matrix. Furthermore, the values for the sphericity 
are in the range $1 \leq \gamma \leq \pdim$. This can be seen by applying
the Cauchy-Schwartz inequality:
\[
	\tr(\M) ^2
	= \Bigg( \sum_{i=1}^\pdim \lambda_i \cdot 1 \Bigg)^2
	\leq \pdim \cdot \sum_{i=1}^\pdim \lambda_i^2
	= \pdim \tr(\M^2).
\]
By dividing the right-hand side of the equation by the left-hand side, we
have $\gamma \geq 1$ with equality if and only if $\M= c \I$ for some $c>0$.
Furthermore, the upper bound $\gamma = \pdim$ is achieved for rank one
matrices, in which case $\M$  has only one non-zero eigenvalue.

The scale and sphericity, $\eta$ and $\gamma$, are elemental in
our developments. As is shown in Theorem~\ref{th:beta0_ELL}, the optimal
shrinkage parameter pair $(\al_o,\be_o)$ for a given elliptical distribution
depends on the true covariance matrix $\M$ only through $\eta$ and $\gamma$.
Simple plug-in estimates of $(\al_o,\be_o)$ can then be obtained by 
replacing $(\eta, \gamma)$ with their estimates. If the elliptical
distribution is unknown an additional elliptical kurtosis parameter needs to
be estimated.

The next theorem provides the expressions for the oracle shrinkage
parameters in the case of sampling from an unspecified $\pdim$-variate
distribution with finite fourth order moments.  Write $\MSE(\S) = \E \Big[
	\big\|   \S - \M \big \|^2_\Fr \Big]$ for the mean squared
error (MSE) and 
\[
	\mathrm{NMSE}(\S) = \frac{\E \big[ \| \S - \M \|_\Fr^2 \big]}
							 {\| \M \|^2_\Fr}.
\]
for the \emph{normalized MSE}. 

\begin{theorem} \label{th:oracle_LW}
Let $\x_1, \ldots, \x_n$ denote an i.i.d.\ random sample from any
$\pdim$-variate distribution with finite fourth order moments, mean vector
$\bmu$, and covariance matrix $\M$. Then, the oracle shrinkage parameters
in~\eqref{eq:RegSCM_oracle} are
\begin{align}
	\be_o
	&= \frac{\pdim (\gamma-1) \eta^2}{\E \big[\tr \big(\S^2\big) \big]
		- \pdim \eta^2} \label{eq:beta0ver2} \\ 
	&= \dfrac{(\gamma-1)}{(\gamma-1) + \gamma \cdot \mathrm{NMSE}(\S)} \label{eq:beta0ver2b}
\end{align}
and 
\begin{align} \label{eq:alpha0}
	\al_o = (1 - \be_o) \eta,
\end{align}
where $\eta$ and $\gamma$ are defined in~\eqref{eq:scale}. Furthermore, the
optimal $\be_o$ is always in the range $[0,1)$ and the value of the MSE
at the optimum is
\begin{equation} \label{eq:MSEopt}
    \MSE(\S_{\al_o,\be_o}) 
    =  ( 1- \be_o) \| \M - \eta \I \|^2_{\Fr}.
\end{equation}
\end{theorem}

\begin{proof}
It was shown in~\cite[Theorem~2.1]{ledoit_wolf:2004} that
\begin{align} \label{eq:beta0ver2_apu}
	\be_o &=
	\frac{\| \M - \eta \I \|^2_\Fr}{\| \M - \eta \I \|^2_\Fr
	+ \MSE(\S)}
\end{align}
and $\al_o = (1 - \be_o) \eta$. Although, the result was shown assuming
$\bmu=\bo 0$ is known and for $\S_{0} = (1/n) \sum_{i=1}^\ndim \x_i
\x_i^\top$, this result transfers to the non-centered case as the derivation
only assumes that $\E[\S_0]=\M$ which applies to $\S$ as well. Note that,
$\be_o = 1$ only if $\S = \M$, and thus $\MSE(\S)=0$, which has zero
probability when sampling from a continuous distribution. The form of
$\be_o$ in~\eqref{eq:beta0ver2_apu} therefore implies that $\be_o \in
[0,1)$. We now show that~\eqref{eq:beta0ver2_apu} can be expressed in the
	form~\eqref{eq:beta0ver2}. 
	
First, we write
\begin{align*}
	a_1
	&= \MSE(\S) 
	= \E \big[ \| \Mor - \M \|^2_{\Fr} \big] \\
	&= \E \big[\tr  \big(\Mor^2\big) \big] 	- 2 \E \big[ \tr \big(\Mor \M\big)\big] + \tr\big(\M^2\big) \\
	&= \E \big[\tr \big(\Mor^2\big) \big]
	- \tr(\M^2),
	\numberthis
	\label{eq:a1}
\end{align*}
where we used that $\E[\tr(\Mor \M)] = \Tr(\E[\Mor]\M) = \tr(\M^2)$.
The numerator of $\be_o$ in~\eqref{eq:beta0ver2_apu} is
\begin{align}
	a_2 &= \| \M - \eta \I \|^2_\Fr
	= \tr(\M^2) - (1/\pdim) \tr(\M)^2 \notag \\
	&= \pdim( \etatwo - \eta^2) = \pdim(\gamma-1) \eta^2,
	\label{eq:a2}
\end{align}
where we denote $\etatwo = \tr(\M^2)/\pdim$. 
This shows that the denominator of $\be_o$ is
$a_1 + a_2
= \E \big[\tr\! \big(\Mor^2\big) \big] - (1/\pdim) \tr(\M)^2
= \E \big[\tr\! \big(\Mor^2\big) \big] - \pdim \eta^2$.
These expressions for the numerator and the denominator of $\beta_o$ yield
the assertion~\eqref{eq:beta0ver2} for $\be_o$.
Substituting~\eqref{eq:a2} into~\eqref{eq:beta0ver2_apu} and multiplying
    both the numerator and denominator by $1/(p\eta^2)$ gives~\eqref{eq:beta0ver2b}.

Next, we derive the expression for the MSE of the RSCM $\S_{\al,\be}$. 
By using the variance and bias decomposition of the MSE, we have
\begin{align*}
	\MSE(\S_{\al,\be})
    &= \tr(\var(\veco{\S_{\al,\be}})) + \|\E[\S_{\al,\be}] -
    \M \|_\Fr^2
    \\
    &= \beta^2\tr(\var(\veco{\S})) + \|\alpha \I - (1-\beta) \M \|_\Fr^2
    \\
    &= \beta^2 \MSE(\S) + \|\alpha \I - (1-\beta) \M \|_\Fr^2.
\end{align*}
	We used the fact that from the unbiasedness of $\S$ it  follows that
	$\MSE(\S) = \tr(\var(\veco{\S})) = a_1$. At the optimum, we have $\be_o
	a_1 = (1-\be_o) a_2$, which can be seen from~\eqref{eq:beta0ver2_apu},
	and $\al_o = (1-\be_o)\eta$. The MSE at the optimum is therefore
\begin{align*}
    \MSE(\S_{(1-\be_o)\eta,\be_o})
    &= \be_o^2 \MSE(\S) + (1-\be_o)^2\| \M - \eta\I\|_\Fr^2
    \\
    &= \be_o(1-\be_o) a_2 + (1-\be_o)^2 a_2
    \\
    &= (1-\be_o) a_2,
\end{align*}
which concludes the proof.
\end{proof}

Theorem~\ref{th:oracle_LW} has important implications. First, since $\al_o =
(1-\be_o)\eta$ is determined by the value of $\be_o \in [0,1)$, the optimal
	RSCM can be expressed as
\[
	\S_{\al_o,\be_o}
	= \be_o
	\S
	+ (1-\be_o)\eta \I.
\]
The scale $\eta$ can be estimated with 
\begin{equation}
	\hat \eta=\frac{\tr(\S)}{\pdim},
	\label{eq:etahat}
\end{equation}
which is a consistent estimator both in the conventional (fixed $\pdim$) and
the RMT asymptotic regime. Therefore, the estimator of $\al_o$ is simply
$\hat \al_o = (1-\hat \be_o) \hat \eta$, and we can focus on finding an
estimator $\hat \be_o$ of $\be_o$. 

This is the approach also taken by Ledoit and Wolf~\cite{ledoit_wolf:2004}
who develop an estimator $\hat \be_o^{\textup{LW}}$ that converges to
$\be_o$ in~\eqref{eq:beta0ver2} under the RMT regime (R1) and some mild
technical assumptions when sampling from a distribution $\x \sim F$ with
finite 8th order moments. The estimate of $\al_o$ is then $\hat
\al_o^{\textup{LW}}	= (1-\hat\be_o^{\textup{LW}})\hat \eta$. The RSCM based
on the shrinkage parameter pair $(\hat \al_o^{\textup{LW}},  \hat
\be_o^{\textup{LW}})$ of \cite{ledoit_wolf:2004} is referred hereafter as
the \emph{LW-RSCM} estimator.


\section{Optimal oracle shrinkage parameters: the elliptical
case}\label{sec:optim2}

We now derive the optimal oracle shrinkage parameters for the case in which
the data can be assumed elliptically distributed. For a review of elliptical
distributions, see~\cite{muirhead:1982,fang1990symmetric},
and~\cite{ollila2012complex}. 

The probability density function (p.d.f.) of an elliptically distributed
random vector $\x \sim \mathcal E_\pdim(\bmu,\M,g)$ is
\[
	f(\x)
	= C_{\pdim,g} |\M|^{-1/2} g\big((\x - \bmu)^\top \M^{-1} (\x-\bmu)\big),
 \]
where $\E[\x]=\bmu$ is the mean vector, $\M$ is the positive definite
covariance matrix, $g: [0,\infty) \to [0,\infty)$ is the \emph{density
generator}, which is a fixed function that is independent of $\x$, $\bmu$
and $\M$, and $C_{\pdim,g}$ is a normalizing constant ensuring that $f(\x)$
integrates to 1. Here, we let $g$ to be defined so that $\M$ represents the
covariance matrix of $\x$, which means that $\int_0^\infty t^{\dim/2} \,
g(t) \mathrm{d} t = \pdim$.  The functional form of the density generator
$g$ determines the elliptical distribution. For example, the multivariate
normal (MVN) distribution, denoted $\x \sim \mathcal N_\pdim(\bmu,\M)$, is
obtained when $g(t)=\exp(-t/2)$. As in Theorem~\ref{th:oracle_LW}, we assume
that the elliptical population possesses finite fourth order moments.
Technically, this implies that
\begin{equation} \label{eq:finite_4mom}
	\int_0^\infty t^{\dim/2+1} \, g(t) \mathrm{d} t < \infty.
\end{equation}
For example, the MVN and the multivariate $t$-distribution with degrees of
freedom $\nu > 4$ all verify the above condition.

The \emph{kurtosis} of a  random variable $x$ is defined as
\[
	\mathrm{kurt}(x) = \frac{\E[ (x-\mu)^4]}{(\E[ (x - \mu)^2])^2} - 3 ,
\]
where $\mu = \E[x]$. The \emph{elliptical kurtosis
parameter}~\cite{muirhead:1982} $\kappa$ of a random vector
$\x=(x_1,\ldots,x_\pdim)^\top \sim \mathcal E_\pdim(\bmu,\M,g)$ is defined
as
\begin{equation}\label{eq:kappa}
	\kappa = \frac{\E[r^4]}{\pdim(\pdim+2)} - 1
	= \frac 1 3 \cdot \mathrm{kurt}(x_i),
\end{equation}
where $r$ is the \emph{generating variate} or \emph{second order modular
variate} of the elliptical distribution, which is defined as the square-root
of the quadratic form $r^2= (\x-\bmu)^\top \M^{-1} (\x-\bmu)$.  Above
$\mathrm{kurt}(x_i)$ denotes the kurtosis of (any) marginal variable $x_i$.
The elliptical kurtosis shares properties similar to the kurtosis of a real
random variable. Especially, if $\x \sim \mathcal N_\pdim(\bmu,\M)$, then
$\ka=0$. This is obvious since the marginal distributions are Gaussian and
hence $\kappa = (1/3) \, \mathrm{kurt}(x_i)=0$. Another way to derive this
is by noting that the quadratic form $r^2$ has a chi-squared distribution
with $\pdim$ degrees of freedom, i.e., $r^2 \sim \chi^2_\pdim$, and hence
$\E[r^4] = \pdim(\pdim+2)$.

The importance of the elliptical kurtosis parameter $\kappa$ is due to the
fact that the $\pdim^2 \times \pdim^2$ covariance matrix of $\ve(\S)$
depends on the underlying elliptical distribution $g$ only through $\kappa$.
This result is  established in Theorem~\ref{th:covS}. 

We will utilize the following matrix decomposition in our proofs. 
Let $\X= ( \x_1 \cdots \x_n )^\top$ denote the $\ndim \times \pdim$ data matrix 
with $i$th transposed observation as its row vector. Then, the SCM can be
written as
\[
\S = \frac{1}{n-1}\X^\top \H \X
\]
where $\H$ is the centering matrix. 

\begin{theorem}\label{th:covS}
Let $\x_1,\ldots, \x_\ndim$ denote an i.i.d.\ random sample from an
elliptical distribution with finite fourth order moments, mean vector
$\bmu$, and covariance matrix $\M$. Then,
\begin{gather}
	\var(\ve(\S)) = \notag \\ 
	\Big(\frac{1}{\ndim-1} + \frac{\ka}{\ndim} \Big)(\I + \commat) (\M \otimes \M)
	+ \frac{\ka}{\ndim} \ve(\M) \ve(\M)^\top.
	\label{eq:cov_vecS}
\end{gather}
\end{theorem}

\begin{proof}
For elliptically distributed observations $\{\x_i \}_{i=1}^\ndim \iidsim
    \mathcal{E}_p (\boldsymbol{\mu},\M,g)$, we have the following stochastic
    decomposition $\x_i =_d {\M}^{1/2} \z_i + \bmu$, where $\z_i
    \sim \mathcal{E}_p(\boldsymbol{0},\I,g)$.  Let $\Z= ( \z_1 \cdots \z_n
    )^\top$ denote the $\ndim \times \pdim$ data matrix collecting the
    random vectors $\z_i$ as  its row vectors.  Then, the stochastic
    decomposition implies that 
\[
    \X^\top \H \X  =_d  \M^{1/2}\Z^\top \H \Z\M^{1/2}. 
\]
Hence, 
\begin{align*}
	&\var\left(\veco{\S}\right)
	= \left(
	\var\left(\frac{1}{n-1}\veco{\M^{1/2}\Z^\top \H \Z\M^{1/2}}
	\right)\right)
	\\
	&= (\M^{1/2}\otimes \M^{1/2})\var\left(\frac{1}{n-1}\veco{\Z^\top \H
	\Z}\right) (\M^{1/2}\otimes \M^{1/2}).
\numberthis
\label{eq:covvecSwithZtHZ}
\end{align*}
Since the matrix $\Z^\top \H \Z$ is radially distributed, we can
apply~\cite{tyler1982radial}, which states
\begin{equation}
	\var\left(\frac{1}{n-1}\veco{\Z^\top\H\Z}\right)
	= \tau_1(\I + \commat) + \tau_2 \veco{\I} \veco{\I}^\top,
\label{eq:covvecZtHZ}
\end{equation}
where the parameters $\tau_1$ and $\tau_2$ correspond to the variance of
any off-diagonal element and the covariance of any two diagonal elements of
the matrix $\frac{1}{n-1}\Z^\top \H \Z$, respectively.

We will first derive the expression for $\tau_1$. For $q \neq r$, it holds
that
\begin{align*}
	{(n-1)}^2\tau_1
	&= \var\left({(\Z\e_q)}^\top \H (\Z \e_r)\right)\\
	&= \var\left(\tr\left(\H (\Z \e_r){(\Z\e_q)}^\top\right)\right)\\
	&= \var\left(\veco{\H}^\top\veco{(\Z\e_r){(\Z\e_q)}^\top}\right)\\
	&= \veco{\H}^\top
	   \var\left(\veco{(\Z \e_r){(\Z\e_q)}^\top}\right)
	   \veco{\H}.
\end{align*}
Next we recall that  $\z_i \sim\ \mathcal{E}_p(\boldsymbol{0},\I,g)$ has a
stochastic representation ({\it cf.} \cite[Theorem 2.9]{fang1990symmetric}) 
$\z_i =_d r_i \u_i$, where $r_i$ is the generating variate with a density
$f(r) = C \cdot r^{p-1} g(r^2)$ (where $C$ is normalizing constant)  and
\mbox{$\u_i = {(u_{i1},u_{i2},\ldots,u_{ip})}^\top$} is uniformly
distributed on the unit hypersphere $\mathcal S^{p-1}=\{ \x \in \R^\pdim
\, : \, \x^\top \x = 1 \}$ and $r_i$ is independent of $\u_i$.  Using
this  stochastic representation for $\z_i$, we can write $\Z\e_q = {(r_1
u_{1q}, r_2 u_{2q}, \ldots, r_n u_{nq})}^\top$, The $kl$th element of
the $ij$th block (i.e., the $ijkl$th element) of the $n^2 \times n^2$
matrix $\var\big(\mathrm{vec}\big((\Z \e_r){(\Z\e_q)}^\top\big)\big)$
can then be written as
\begin{align*}
	&\cov\left((\Z\e_r)_k (\Z\e_q)_i , (\Z\e_r)_l (\Z\e_q)_j\right)
	 = \\
	 &\quad\E\left[ r_{k} u_{kr} \cdot r_{i} u_{iq} \cdot r_{l} u_{lr} \cdot
	 r_{j} u_{jq} \right]
	 \\&\quad
	 -\E\left[ r_{k} u_{kr} \cdot r_{i} u_{iq}\right]
	 \E\left[r_{l} u_{lr} \cdot r_{j} u_{jq} \right].
\end{align*}
Using the following identities for $\forall i,j$ and $q\neq
r$ ({\it cf.} \cite[Section~3.1]{fang1990symmetric}) : 
\begin{align*}
	\E\left[u_{iq} u_{jr}\right]      &= 0,
	&\E\left[u_{iq}^2 u_{ir}^2\right] &= \frac{1}{p(p+2)},
	\\
	\E\left[u_{iq}^2\right]           &= \frac{1}{p},
	&\E\left[u_{iq}^4\right]          &= \frac{3}{p(p+2)},
	\\
	\E\left[r_i^2\right]              &= p,
	&\E\left[r_i^4\right]             &= (1+\kappa)p(p+2),
\end{align*}
where $3\kappa = \text{kurt}(z_{iq}) = \text{kurt}(x_{iq})$, we find that
the only non-zero elements of $\var\big(\mathrm{vec}\big((\Z
\e_r){(\Z\e_q)}^\top\big)\big)$ correspond to
\begin{align*}
	\E[r_i^4] \E[u_{ir}^2 u_{iq}^2] &= 1+\kappa
	&\text{for}~i&=j=k=l,~\textup{and}
	\\
	\E[r_{i}^2]\E[r_{k}^2] \E[u_{ir}^2]\E[u_{kq}^2] &= 1
	&\text{for}~i&=j, k=l, i\neq k.
\end{align*}
This implies that
\begin{equation}
	\var\left(\veco{(\Z\e_r)(\Z \e_q)^\top}\right)
	= \I + \kappa \sum_{i=1}^n \e_i \e_i^\top \otimes \e_i \e_i^\top.
\end{equation}
Hence, we can write $ \tau_1$  as
\begin{align*}
	\tau_1 &=
	\frac{1}{{(n-1)}^2} \veco{\H}^\top
	\left(
		\I + \kappa \sum_{i=1}^n \e_i \e_i^\top \otimes \e_i \e_i^\top
	\right)
	\veco{\H}\\
	&=
	\frac{1}{n-1} + \frac{\kappa}{n},
	\numberthis
	\label{eq:tau1}
\end{align*}
where we used $\veco{\H}^\top \veco{\H} = n - 1$ and
\begin{align*}
	\sum_{i=1}^n \veco{\H}^\top (\e_i \e_i^\top \otimes \e_i \e_i^\top)
	\veco{\H} = \sum_{i=1}^n h_{ii}^2 = \frac{{(n-1)}^2}{n}.
\end{align*}

Next, we find the expression for $\tau_2$. For $q\neq r$, we have
\begin{align*}
	{(n-1)}^2\tau_2
	&=
	\cov\left((\Z\e_q)^\top \H (\Z \e_q), (\Z\e_r)^\top \H (\Z \e_r)\right)
	\\
	&=
	\E\left[(\Z\e_q)^\top \H (\Z \e_q) (\Z\e_r)^\top \H (\Z \e_r)\right]
	\\&\quad
	- \E[(\Z\e_q)^\top \H (\Z \e_q)] \E[(\Z\e_r)^\top \H (\Z \e_r)].
\end{align*}
By using basic algebraic properties of the trace and the vectorization
transform and noting that $\E[(\Z \e_r)(\Z\e_r)^\top] = \I$, we arrive at
the form
\begin{align*}
	&{(n-1)}^2\tau_2
	=
	\\
	&\tr\left((\H \otimes \H)\E\left[
		\veco{(\Z \e_q) (\Z\e_r)^\top}\veco{(\Z\e_q)(\Z \e_r)^\top}^\top
		\right]\right)
	\\
	&\quad - \tr(\H)^2.
\end{align*}
The expression involving the expectation is equal to
$\var\left(\veco{(\Z\e_q)(\Z\e_r)^\top}\right)$, which implies
\begin{gather*}
	{(n-1)}^2\tau_2 
	= \\
	\tr \left\{(\H \otimes \H)
	\left(
		\I + \kappa \sum_{i=1}^n \e_i \e_i^\top \otimes \e_i \e_i^\top
	\right) \right\}
	-
	\tr(\H)^2.
\end{gather*}
By noting that $\tr(\H \otimes \H) = \tr(\H)^2$ and
\begin{align*}
	\sum_{i=1}^n
	\tr\left((\H \otimes \H)(\e_i \e_i^\top \otimes \e_i\e_i^\top \right)
	=
	\sum_{i=1}^n h_{ii}^2,
\end{align*}
we find that
\begin{align*}
	\tau_2
	=
	\frac{1}{{(n-1)}^2}
	\kappa \frac{{(n-1)}^2}{n}
	=
	 \frac{\kappa}{n}.
	\numberthis
	\label{eq:tau2}
\end{align*}
By substituting~\eqref{eq:covvecZtHZ} into~\eqref{eq:covvecSwithZtHZ}, and
noticing that 
\begin{align*}
	(\M^{1/2} \otimes \M^{1/2}) \ve(\I) &= \ve(\M)~\text{and}\\
	(\M^{1/2} \otimes \M^{1/2})
	(\I + \commat)
	(\M^{1/2} \otimes \M^{1/2})
	&= (\I + \commat) 
	(\M \otimes \M),
\end{align*}
completes the proof.
\end{proof}

Theorem~\ref{th:covS} reveals that the elliptical kurtosis parameter
$\kappa$ along with the true covariance matrix $\M$ provide a complete
description of the covariances between the elements $s_{ij}$ and $s_{kl}$ of
the SCM $\S=(s_{ij})$. The mathematics underlying Theorem~\ref{th:covS} is
so rich that we are able to relate it to at least three fundamental results
in the field of statistics given below.

First, consider the one-dimensional case, $\pdim=1$, where we have a
univariate sample $x_1, \ldots, x_\ndim$ from a distribution of a random
variable $x \in \R$. Then, the SCM reduces to the unbiased sample variance
$s^2 = (1/(\ndim-1)) \sum_{i=1}^\ndim (x_i - \bar x)^2$ and
equation~\eqref{eq:cov_vecS} reduces to $\var(s^2)$. We can now compute
$\var(s^2)$ using~\eqref{eq:cov_vecS}, which states that 
\begin{align}
	\var(s^2)
	&= \left( \frac{1}{\ndim-1} + \frac{\kappa}{\ndim} \right) 2 \sigma^4
		+ \frac{\ka}{\ndim} \sigma^4 \notag \\
	&= \sigma^4 
    \left(\frac{2}{\ndim - 1} + \frac{\mathrm{kurt}(x)}{\ndim} \right),
	\label{eq:univar}
\end{align}
where we used that $\M \equiv \sigma^2 = \E[ (x- \E[x])^2]$, $\M \otimes \M
= \sigma^4$ and $\kappa = \mathrm{kurt}(x)/3$ due to~\eqref{eq:kappa}.
Hence, we obtained the classic formula for $\var(s^2)$ often encountered in
elementary statistics textbooks. Under the Gaussian distribution,
$\mathrm{kurt}(x)=0$, in which case Theorem~\ref{th:covS} states that
$\var(s^2)= 2\sigma^4/(\ndim-1)$. This is an expected result since
$(\ndim-1)s^2/\sigma^2 = \sum_i (x_i - \bar x)^2/\sigma^2 \sim
\chi^2_{\ndim-1}$.

Secondly, we can connect Theorem~\ref{th:covS} with the well-known
covariance matrix of the Wishart distribution. Let $W_\pdim(m, \bo M)$
denote the Wishart distribution of a random symmetric positive definite
$\pdim \times \pdim$ matrix where $m > \pdim-1$ denotes the degrees of
freedom parameter and $\bo M \in \PDHm{\pdim}$ denotes the scale matrix
parameter of the Wishart distribution. Under the MVN assumption, it is
well-known that $(\ndim-1)\S \sim W_\pdim( \ndim-1, \M)$ and consequently
$\var(\ve(\S))$ has the famous covariance matrix form
\begin{equation}
	\var( \ve(\S)) = \frac{1}{\ndim-1} (\I + \commat) (\M \otimes \M).
	\label{eq:cov_Wishart}
\end{equation}
Suppose now that the elliptical distribution in Theorem~\ref{th:covS} is the
multivariate normal, thus, $\x_1,\ldots,\x_\ndim \iidsim \mathcal
N_\pdim(\bmu, \M)$. Since in this case $\ka=0$, we have
that~\eqref{eq:cov_vecS} reduces to~\eqref{eq:cov_Wishart}.

Lastly, notice that
\[
	\var( \sqrt{n}\ve(\S)) \to (1+\kappa ) (\I + \commat) (\M \otimes \M)
	+ \kappa \, \ve(\M) \ve(\M)^\top
\]
as $\ndim \to \infty$. The right hand side of the previous equation equals
the well-known asymptotic covariance matrix of the limiting normal
distribution of $\sqrt{n}(\ve(\S) - \ve(\M))$ when sampling from an
elliptical distribution $\mathcal E_\pdim(\bmu, \M,g)$ with finite
fourth order moments. This is a famous result in multivariate
statistics~\cite{muirhead:1982}.

In the next Lemma, we derive the MSE and normalized MSE of the SCM.

\begin{lemma}\label{lem:NMSE}
Let $\x_1,\ldots,\x_\ndim$ denote an i.i.d.\ random sample from a $p$-variate 
elliptical distribution with finite fourth order moments, mean $\bmu$, and
covariance matrix $\M$. Then, the MSE and the NMSE of $\S$
are
\begin{align*}
	\MSE(\S) 
		&= \Big( \frac{1}{\ndim-1} + \frac{\ka}{\ndim}
\Big) \tr(\M)^2 + \Big(\frac{1}{\ndim-1} +  \frac{ 2\ka}{\ndim} \Big) \tr(\M^2) \\
	\mathrm{NMSE}(\S) 
	&= \left(1+\frac{p}{\gamma}\right)\Big( \frac{1}{\ndim-1} + \frac{\ka}{\ndim}
\Big)+ \frac{\ka}{\ndim}
\end{align*}
where  $\gamma$ and $\ka$ are defined in~\eqref{eq:scale} and \eqref{eq:kappa}, 
respectively.
\end{lemma}

\begin{proof}
Since $\S$ is unbiased, i.e., $\E[\S] = \M$, it holds that
\begin{align*}
	\mathrm{MSE}(\S) 
	&= 
	\E\left[\|\S - \M\|_\Fr^2\right]
	\\
	&= 
	\E\left[ \ve(\S-\M)^\top \ve(\S-\M) \right]
	\\
	&= 
	\tr\left(\E\left[
		\ve\left(\S - \E[\S]\right)
		\ve\left(\S - \E[\S]\right)^\top
			\right]\right)
	\\
	&=
	\tr \left( \var( \ve(\S)) \right),
	\numberthis
	\label{eq:lem:NMSE:apu1}
\end{align*}
Then we substitute the expression stated in~\eqref{eq:cov_vecS} for
$\var(\ve(\S))$ into equation \eqref{eq:lem:NMSE:apu1} and use the
following identities 
\begin{align*}
	\tr(\M \otimes \M) &= \tr(\M)^2, 
	\\
	\tr \left(\ve(\M) \ve(\M)^\top \right) &= \tr(\M^2),~\text{and}
	\\
	\tr\big( \commat (\M \otimes \M) \big) &= \tr(\M^2)
\end{align*}
where the last identity follows from
\begin{align*}
	\tr\big( \commat (\M \otimes \M) \big) 
	&= 
	\sum_{i,j} \tr\big( \left(\e_i \e_j^\top \otimes \e_j \e_i^\top \right)
	(\M \otimes \M) \big)
	\\
	&= 
	\sum_{i,j} \tr\left(\e_j^\top \M \e_i \cdot \e_i^\top \M \e_j \right)
	\\
	&= 
	\tr(\M^2),
\end{align*}
to obtain the stated expression of the $\MSE(\S)$. The expression for the
$\mathrm{NMSE}(\S)$ is obtained by dividing the $\MSE(\S)$ by
$\|\M\|_\Fr^2 = \tr(\M^2)$.
\end{proof}

The next theorem states that the oracle parameters derived in
Theorem~\ref{th:oracle_LW} can be written in a much simpler form when
sampling from an elliptically symmetric distribution.

\begin{theorem}\label{th:beta0_ELL}
Let $\x_1,\ldots,\x_\ndim$ denote an i.i.d.\ random sample from an
elliptical distribution with finite fourth order moments, mean $\bmu$, and
covariance matrix $\M$. Then the oracle parameters $(\al_o,\be_o)$ that
minimize the MSE are
\begin{align*}
    \be_o^{\textup{Ell}}
	&=  \dfrac{(\gamma-1)}
	{( \gamma - 1) + \ka (2 \gamma  + p)/\ndim + (\gamma  + p)/(n-1) }, 
\end{align*}
and $\al_o^{\textup{Ell}} = (1-\be_o^{\textup{Ell}}) \eta$, where the
parameters $\eta$, $\gamma$ and $\kappa$ are defined in~\eqref{eq:scale}
and~\eqref{eq:kappa}, respectively. 
\end{theorem}

\begin{proof} 
Follows from  \eqref{eq:beta0ver2b} and Lemma~\ref{lem:NMSE}.
\end{proof}

It is not surprising that $\beta_o$, and hence also $\al_o$, depend on the
density generator $g$ of the elliptical distribution only via the elliptical
kurtosis parameter $\kappa$. Specifying the elliptical distribution also
specifies the value of $\kappa$. For example, when sampling from the
Gaussian distribution, the elliptical kurtosis parameter is $\kappa=0$ and 
$\be_o^{\textup{Ell}}$ in Theorem~\ref{th:beta0_ELL} reduces to
\begin{equation} \label{eq:beta0Gau}
    \be_o^{\textup{Gau}}
	= \frac{(\gamma-1)}{(\gamma-1) + (\gamma + \pdim)/(\ndim-1)}.
\end{equation}
Consequently, an estimator of $\be_o^{\textup{Gau}}$ is obtained by
substituting an estimator $\hat \gamma$ in place of $\gamma$
in~\eqref{eq:beta0Gau}. Recall that an estimator of $\al_o^{\textup{Gau}}$
is then obtained as $\hat \al_o^{\textup{Gau}} = (1-\hat
\be_o^{\textup{Gau}}) \tr(\S)/\pdim$. 

Since in this paper we do not assume any particular elliptical distribution,
we need to find an estimator $\hat \kappa$ of the elliptical kurtosis
parameter $\kappa$ as well. Naturally, if the assumption on multivariate
normality of the data is valid, then~\eqref{eq:beta0Gau} should be used for
estimating the optimal oracle value.

When the mean vector of the population is known, the unbiased SCM is
$\S=\frac 1 \ndim \sum_{i=1}^\ndim \x_i \x_i^\top$ (as one can assume
without loss of generality that $\bom \mu = \bo 0$). In this case the
optimal shrinkage parameter $\be_o$ of the RSCM stated in
Theorem~\ref{th:beta0_ELL} remains unchanged apart from the  last term in
the denominator of $\be_o$, that is, $(\gamma + \pdim)/(n-1)$ is replaced by
$(\gamma + \pdim)/n$. This centered case was addressed in
\cite{ollila2017optimal}.

\section{Estimation of the oracle parameters} \label{sec:estim}

In this section, we develop estimators $\hat \gamma$ and $\hat \kappa$ of
the unknown parameters $\kappa$ and $\gamma$ that determine the shrinkage
parameter $\beta_o$ (cf.  Theorem~\ref{th:beta0_ELL}). These are used to
obtain a \emph{plug-in estimators} of the shrinkage parameters as 
\begin{align*}
	\hat \beta_o^{\textup{Ell}}
	&=  \dfrac{( \hat \gamma-1)}
	{(   \hat \gamma - 1) + \hat \ka (2  \hat \gamma  + p)/\ndim + ( \hat \gamma  + p)/(n-1) }, 
	\\
	\hat \al_o^{\textup{Ell}} &= (1 - \hat \be_o^{\textup{Ell}}) \hat\eta.
\end{align*}

Next, we will address how to estimate the needed statistical parameters.
First, we will address the estimation of  $\kappa$. Regarding
$\gamma$, we found two different well performing estimators, and hence, we
will address its estimation last.

A natural estimate of $\kappa$ is the conventional sample average
\begin{equation}\label{eq:kappahat}
	\hat \kappa = \max \Bigg( - \frac{2}{p+2} \, , \,
	\frac{1}{3 \pdim} \sum_{j=1}^\pdim \hat K_j \Bigg),
\end{equation}
where $\hat K_j$ is an estimate of the kurtosis of the $j$th variable and 
defined as
\[
	\hat K_j = \frac{\ndim-1}{(\ndim-2)(\ndim-3)}
	\left( (\ndim +1) \hat k_j + 6 \right).
\]
Here $\hat k_j = m^{(4)}_j/ \big(m_j^{(2)} \big)^2 - 3$ denotes the
conventional sample estimate of the kurtosis of the $j$th variable, where
$m^{(q)}_j = \frac 1 \ndim \sum_{i=1}^\ndim (x_{ij} - \bar x_j)^q$ denotes
the $q$th order sample moment. The estimate $\hat K_j$ is a commonly used
estimate of the kurtosis which is based on the relationship between the
kurtosis and the cumulants of the distribution~\cite{joanes1998comparing}.
It corrects for the bias of the conventional sample kurtosis $\hat k_j$. To
ensure that the final estimate $\hat \kappa$ does not go below the
theoretical lower bound of $-2/(p+2)$~\cite{bentler1986greatest}, a maximum
constraint is used in~\eqref{eq:kappahat}. The constructed estimate of
$\kappa$ is consistent both in the conventional and the RMT regime. 

Note that, if the estimates $(\hat \gamma, \hat \kappa)$ are restricted to
be within their theoretical ranges, i.e., $1 \leq \hat \gamma \leq p$ and
$\hat\kappa \geq -2/(p+2)$, then it is straightforward to verify that the
plug-in estimator satisfies $\hat \beta_o^\textup{Ell} \in [0,1)$. 

In the following subsections, we consider two options for estimating the
sphericity $\gamma$ under the RMT regime. We denote the estimators
by $\hat \gamma^{\text{Ell1}}$ and $\hat \gamma^{\text{Ell2}}$. Both
estimators have their own benefits and disadvantages. The first estimator,
$\hat \gamma^{\text{Ell1}}$, enjoys statistical robustness with respect to
heavier-tailed distributions. The second estimator, $\hat
\gamma^{\text{Ell2}}$, is computationally more efficient and can easily be
used and tuned for very high-dimensional set-ups such as microarray studies
where $\pdim$ is often tens of thousands but $\ndim$ is of few tens
\cite{tabassum2018compressive}. It is also highly efficient under
Gaussianity, or for mild departures from Gaussianity. Its obvious
disadvantage is that it is not very efficient for heavier-tailed elliptical
distributions.

\subsection{Ell1-RSCM estimator}

The first estimator of the sphericity $\gamma$,
uses the \emph{sample spatial sign covariance matrix},  defined as
\begin{equation} \label{eq:sgn}
	\S_{\mathrm{sgn}} = \frac{1}{\ndim} \sum_{i=1}^\ndim
	\frac{(\x_i - \hat \bmu) (\x_i - \hat \bmu)^\top}
	     {\| \x_i - \hat \bmu\|^2},
\end{equation}
where $\hat \bmu = \arg \min_{\bmu} \sum_{i=1}^{n} \| \x_{i} - \bmu \|$ is
the \emph{sample spatial median}~\cite{brown1983statistical}. The sample
sign covariance matrix is well-known to be highly robust although it is not
a consistent estimator of the covariance
matrix~\cite{croux2002sign,magyar2014asymptotic}. Namely, it does provide
consistent estimators of the eigenvectors of the covariance matrix but not
of the eigenvalues. 

Consider an estimator of the form, 
\begin{align*} 
	\hat \gamma^{\text{Ell1}*}  & =  \frac{\ndim}{ \ndim -1}  \, \left(
	\pdim \tr \big(\S_{\mathrm{sgn}}^2 \big) - \frac{\pdim}{\ndim}  \right)
	\\ 
	& =  \frac{ p }{n(n-1)}  \sum_{i \neq j} (\v_i^\top \v_j)^2  \\  
	& = p  \, \mathrm{ave}_{i\neq j}  \{ \cos^2(\sphericalangle(\x_i, \x_j) ) \}
\end{align*} 
where $\mathrm{ave}_{i,j}$ denotes arithmetic average over indices, $i,j \in
\{ 1, \ldots,n \}$, $i \neq j$, and $\v_i =  (\x_i - \hat \bmu)/\| \x_i  -
\hat \bmu\|$. 

In~\cite[Lemma~4.1]{zhang2016automatic}  it was shown that $\hat
\gamma^{\text{Ell1}*} $ is a consistent estimator of $\gamma$  when sampling
from a centered elliptical distribution $\mathcal E_\pdim(\bo 0,\M,g)$ under
the RMT regime~(R1) and when the eigenvalues of $\M$ converge to a fixed
spectrum such that
\begin{equation}
\begin{aligned}
	\tr(\M^i)/p \text{ has a finite and positive limit } \\ \text{for}~i=1,2,3,4~\text{as}~ p \to \infty.
\end{aligned}  \tag{R2}
	\end{equation}
In their paper, it was assumed that the location (symmetry center) is known
to be zero, which is why they do not have the centering of the samples by
the sample spatial median in~\eqref{eq:sgn}. We also remark that  our
estimator $\hat \gamma^{\text{Ell1}*}$ differs from
\cite[Lemma~4.1]{zhang2016automatic} in that we scale their estimator by $
\frac{\ndim}{ \ndim -1}$. This scaling is used for correcting bias for small
samples and needed to ensure that $\E[\hat \gamma^{\text{Ell1*}}]  \in
[1,p]$. In order to guarantee that the estimate remains inside the
valid interval $[1,p]$, as a final estimator, we use 
\begin{equation} \label{eq:gammahatEll1}
	\hat \gamma^{\text{Ell1}} =\min(p,\max(1,\hat\gamma^{\text{Ell1*}})).
\end{equation}
We can now define the \emph{Ell1-RSCM} estimator as the RSCM based on the
estimators of the optimal shrinkage parameters using the plugin estimates
$\hat\eta$ of~\eqref{eq:etahat}, $\hat\kappa$ of~\eqref{eq:kappahat} and
$\hat \gamma^{\text{Ell1}}$ of~\eqref{eq:gammahatEll1}.

\subsection{Ell2-RSCM estimator}
In order to develop the second estimator $\hat\gamma^{\text{Ell2}}$ of
$\gamma$, we need to find the values of $\E[\tr(\S^2)]$ and $\E[\tr(\S)^2]$,
which are given in the following Lemma~\ref{lem:trace_S}.

\begin{lemma} \label{lem:trace_S}
Let $\x_1,\ldots, \x_\ndim$ denote an i.i.d.\ random sample from an
elliptical distribution with finite fourth order moments, mean vector
$\bmu$, and covariance matrix $\M$. Then,
\begin{align*}
	\E&[\tr(\S^2)] = \\ 
	&\left( \frac{1}{\ndim-1} + \frac{\ka}{\ndim} \right) \tr(\M)^2 
	+ \left(1+\frac{1}{\ndim-1} +  \frac{ 2\ka}{\ndim} \right) \tr(\M^2) 
\end{align*}
and
\begin{align*}
	\E[\tr(\S)^2] 
	&= \bigg(1+ \frac{\ka}{\ndim} \bigg)\tr(\M)^2  +  
	2 \left( \frac{1}{\ndim-1} + \frac{\ka}{\ndim} \right) \tr(\M^2) 
\end{align*}
\end{lemma}

\begin{proof} 
The first statement for $\E[\tr(\S^2)]$ follows from Lemma~\ref{lem:NMSE} by
noting that $\E[\tr(\S^2)] = \MSE(\S) + \tr(\M^2)$, which was shown
in~\eqref{eq:a1}.

Regarding the second statement, we first write
\begin{align*}
	\E\left[\tr(\S)^2\right]
	&= \E\Big[\sum_i s_{ii} \sum_j s_{jj}\Big]
	= \sum_{i,j}\E\left[s_{ii}s_{jj}\right]
	\\
	&= \sum_{i,j} \left(\cov\left(s_{ii},s_{jj}\right) +
	\E\left[s_{ii}\right]\E\left[s_{jj}\right]\right).
\end{align*}
Here, the covariance of $s_{ii}$ and $s_{jj}$ is the $ij$th element of the
$ij$th block of $\var\left(\veco{\S}\right)$ in~\eqref{eq:cov_vecS} since
\begin{align*}
	\cov(s_{ii},s_{jj})
	&= \cov(\e_i^\top \S \e_i , \e_j^\top \S \e_j)
	\\
	&= \cov\left((\e_i \otimes \e_i)^\top \ve(\S)  , (\e_j \otimes \e_j)^\top
	\ve(\S)\right)
	\\
	&=
	(\e_i^\top \otimes \e_i^\top)
	\var(\veco{\S})
	(\e_j \otimes \e_j).
\end{align*}
Using the following identities:
\begin{align*}
	(\e_i^\top \otimes \e_i^\top)
	(\M \otimes \M)
	(\e_j \otimes \e_j)
    &=
    \e_i^\top \M \e_j \cdot \e_i^\top \M \e_j,
    \\
	(\e_i^\top \otimes \e_i^\top) \commat(\M \otimes \M) (\e_j \otimes
	\e_j)
    &=
    \e_i^\top \M \e_j \cdot \e_i^\top \M \e_j,
    \\
	(\e_i^\top \otimes \e_i^\top)
	\veco{\M}\veco{\M}^\top
	(\e_j \otimes \e_j)
    &=
	\e_i^\top \M \e_i \cdot \e_j^\top \M \e_j,
\end{align*}
and the fact that $\S$ is unbiased, i.e., $\E[s_{ii}] = \e_i^\top \M \e_i$,
we can write using \eqref{eq:cov_vecS} that 
\begin{align*}
	\E[s_{ii}s_{jj}]
	&= \cov(s_{ii},s_{jj}) + \E[s_{ii}]\E[s_{jj}]
	\\
	&=
    2\tau_1
	(\e_i^\top \M \e_j \cdot \e_i^\top \M \e_j )
	+
    (1+\tau_2)
	(\e_i^\top \M \e_i \cdot \e_j^\top \M \e_j),
\end{align*}
where $\tau_1$ and $\tau_2$ are given in~\eqref{eq:tau1} and
\eqref{eq:tau2}, respectively. By summing all $i$ and $j$, we have
\begin{align*}
	\E[\tr(\S)^2] = \sum_{i,j} \E[s_{ii}s_{jj}]
	&=
    2\tau_1
	\tr(\M^2)
	+
    (1+\tau_2)
	\tr(\M)^2,
\end{align*}
which completes the proof.
\end{proof}

Next, we construct an estimator for $\etatwo = \tr(\M^2)/p$. The natural
plug-in estimate, $ \tr(\S^2)/\pdim$, is not a consistent estimator in the
RMT regime (R1) and (R2). This follows at once from Lemma~\ref{lem:trace_S}
as  it shows that $\tr(\S^2)/\pdim$ is not asymptotically unbiased since 
\begin{align*}
	\lim_{\substack{n,p\to \infty \\ p/n \to c_o}} \frac{\E[\tr(\S^2)]}{\pdim}
	&=	c_0(1+\kappa)\eta^2_o + \etatwo_0, 
\end{align*}
where $\eta_o>0$ and  $\etatwo_o>0$ denote finite limit values of
$\tr(\M)/p$ and $\tr(\M^2)/p$, respectively, as $p\to \infty$. 

In the next Theorem~\ref{th:eta2ell}, a proper estimator
$\hat\etatwo$ of $\etatwo$ under the RMT regime is developed. Theorem~\ref{th:eta2ell}
extends~\cite[Lemma~2.1]{srivastava2005some} to the elliptical case.

\begin{theorem}\label{th:eta2ell} 
	Let $\x_1,\ldots, \x_\ndim$ denote an i.i.d.\ random sample from a $p$-variate 
	elliptical distribution with finite fourth order moments, mean vector
	$\bmu$, and covariance matrix $\M$. Then, an unbiased
	estimate of $ \etatwo = \tr(\M^2)/p$ for any finite $n$ and any
	$p$ is
  \begin{align*}
	  \hat \etatwo & =   b_n \left(\frac{\tr(\S^2)}{p} -  a_n \, \frac{p}{n}
	  \, \left[ \frac{\tr(\S)}{p} \right]^2 \right)   
  \end{align*}
  \vspace{-0.1cm}
where 
\begin{align*}
	a_n &=  \left(\frac{n}{n+\kappa} \right) \left(  \frac{n}{n-1} + \kappa \right)  \\ 
	b_n &= \frac{ (\ka  + n)(n-1)^2}{ (n-2)(3 \kappa (n-1) + n(n+1))}. 
\end{align*}
Furtheremore, under the RMT regime (R1) and (R2), the estimator is
asymptotically unbiased, i.e., $\E[\hat \etatwo] \to \etatwo_o$, where
$\etatwo_o>0$ denotes the finite limit of  $\etatwo$ as $p \to \infty$. 
\end{theorem}

\begin{proof}
Using Lemma~\ref{lem:trace_S} write 
\begin{align*}
	&b_n^{-1} p \E[\hat\etatwo]
	=
	\\
	&\left(\tau_1-\frac{a_n}{n}(1+\tau_2)\right)\tr(\M)^2 
	+ \left(1+\tau_1+\tau_2-2\tau_1 \frac{a_n}{n}\right)\tr(\M^2),
\end{align*} 
where $\tau_1$ and $\tau_2$ are defined in~\eqref{eq:tau1}
and~\eqref{eq:tau2}. By choosing $a_n= n \tau_1/(1+ \tau_2)$, we see
that $b_n = ( 1+\tau_1+\tau_2 - 2\tau_1 a_n/n  \big)^{-1}$. The terms
$a_n$ and $b_n$ can equivalently be expressed in the form given in the
theorem by using the equations for $\tau_1$ and $\tau_2$. The last
statement is a consequence of the fact that $\etatwo$ converges to a
finite limit value as $\pdim \to \infty$ and that $\E[\hat\etatwo] =
\etatwo$.
\end{proof}

Note that $a_n$ and $b_n$ depend on the elliptical distribution via the
elliptical kurtosis parameter $\ka$. Using the estimate of the kurtosis by
defining $\hat a_n = a_n(\hat \kappa)$ and $\hat b_n = b_n(\hat \ka)$ one
obtains an estimate of $\etatwo$ which does not require knowing the
underlying elliptical distribution. Thus based on Theorem~\ref{th:eta2ell},
we propose an estimator of the form 
\begin{equation}\label{eq:hatgamma_ell2}
	\hat \gamma^{\textup{Ell2*}} 
	= \hat b_n \left( \frac{p\tr(\S^2)}{\tr(\S)^2} - \hat a_n c \right).
\end{equation}
Note that, if $n$ is  reasonably large (e.g., $n>100$), then $\hat a_n
\approx 1+  \hat \ka$ and $b_n \approx 1$ and then one may use
\[
	\hat \gamma^{\textup{Ell2}*}  =  \left( \frac{p\tr(\S^2)}{\tr(\S)^2} -
    (1+ \hat \kappa) c \right). 
 \]
In order to guarantee that the estimator remains in the valid interval, $1
\leq \gamma \leq p$, we use 
\begin{equation}\label{eq:hatgamma_ell2_eq2}
   \hat \gamma^{\textup{Ell2}}  
   = \min(p, \max(1, \hat \gamma^{\textup{Ell2*}})).
\end{equation}
as our final estimator. We can now define the \emph{Ell2-RSCM} estimator as
the RSCM, which uses $\hat \eta$ of~\eqref{eq:etahat}, $\hat\kappa$
of~\eqref{eq:kappahat}, and $\hat \gamma^{\text{Ell2}}$
of~\eqref{eq:hatgamma_ell2_eq2} as plug-in estimates for the optimal
shrinkage parameters.

Finally, we wish to note that albeit $\hat \gamma^{\textup{Ell2}}$ does not
require knowledge of the underlying elliptically symmetric distribution of
the data, it is not a robust estimator. This is due to the fact that
$\tr(\S^2)$ contains $4$th order moments of the data, and $8$th order
moments of the elliptically symmetric distribution needs to exists in order
for $\tr(\S^2)$ to be asymptotically normal. Consequently, the Ell2-RSCM
estimator is not well suited for heavier-tailed distributions. 

\subsection{Ell3-RSCM estimator}

Ell3-RSCM is a hybrid of the Ell1-RSCM and Ell2-RSCM estimators. The
Ell3-RSCM will use the estimator which has a smaller estimated sphericity
$\hat\gamma$. Thus, it will always favor more shrinkage over less
shrinkage. This rule can be summarized as: if
$\hat\gamma^{\textup{Ell1}}<\hat\gamma^{\textup{Ell2}}$, then choose
Ell1-RSCM, otherwise choose Ell2-RSCM.

\section{Simulation study} \label{sec:simul}

We conduct a small simulation study to investigate the performance of the
RSCM estimators in terms of their finite sample NMSE.\ Each simulation is
repeated 10000 times and the NMSE is averaged over the Monte-Carlo runs for
each RSCM estimator. The theoretical oracle MSE value derived
in~\eqref{eq:MSEopt} is normalized by $\| \M \|^2_{\Fr}$ and used as a
benchmark for the empirical NMSE values, which is shown in the figures as a
solid black line. The mean vector $\bmu$ is fixed for each MC trial and
generated randomly as $\{\mu_i \}_{i=1}^\pdim \iidsim \mathcal N(0,4)$.

\subsection{AR(1) covariance matrix}\label{sec:AR1}

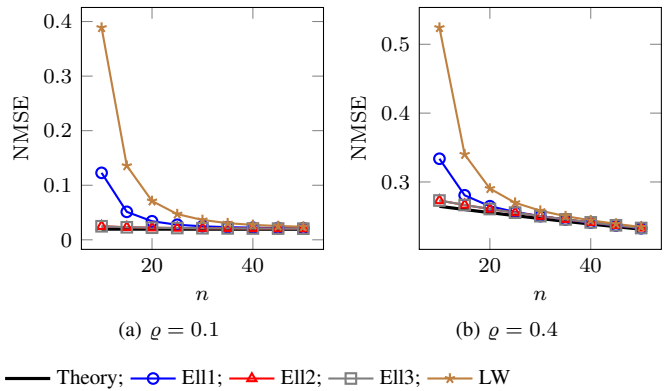
\begin{figure}[t]
	\centering
    \subfloat[$\varrho=0.1$]{%
        \begin{tikzpicture}[scale=1]
            \begin{axis}[mystyle,
                xlabel = {$n$}, 
                ylabel = {$\mathrm{NMSE}$},
                ]
                \pgfplotstableread{AR1MSErho0p1vInf.dat}\loadedtable
                \addplot[black,very thick] table [x=nlist, y=NMSEtheory]{\loadedtable};
                \addplot[blue,mark=o,thick] table [x=nlist, y=NMSEell1]{\loadedtable};
                \addplot[red,mark=triangle,thick] table [x=nlist, y=NMSEell2]{\loadedtable};
                \addplot[gray,mark=square,thick] table [x=nlist, y=NMSEell3]{\loadedtable};
                \addplot[brown,mark=star,thick] table [x=nlist, y=NMSElwe]{\loadedtable};
                \legend{Theory;,Ell1;,Ell2;,Ell3;,LW};
            \end{axis}
        \end{tikzpicture}
    }
	\hfill
    \subfloat[$\varrho=0.4$]{%
        \begin{tikzpicture}[scale=1]
            \begin{axis}[mystyle,
                xlabel = {$n$},
                ylabel = {$\mathrm{NMSE}$},
                ]
                \pgfplotstableread{AR1MSErho0p4vInf.dat}\loadedtable
                \addplot[black,very thick] table [x=nlist, y=NMSEtheory]{\loadedtable};
                \addplot[blue,mark=o,thick] table [x=nlist, y=NMSEell1]{\loadedtable};
                \addplot[red,mark=triangle,thick] table [x=nlist, y=NMSEell2]{\loadedtable};
                \addplot[gray,mark=square,thick] table [x=nlist, y=NMSEell3]{\loadedtable};
                \addplot[brown,mark=star,thick] table [x=nlist, y=NMSElwe]{\loadedtable};
            \end{axis}
        \end{tikzpicture}
    }
    \vspace{2em}
    \caption{AR(1) process: comparison of covariance estimators when $\pdim =
        100$, $\varrho \in \{0.1, 0.4\}$, and the samples are from a Gaussian
    distribution.}\label{fig:AR1_p100}
\end{figure}

In the first experiment, an autoregressive covariance structure is used.
We let $\M$ be the covariance matrix of a Gaussian AR$(1)$ process,
\[
	(\M)_{ij} = \varrho^{|i-j|}, \quad~\text{where}~\varrho \in (0,1).
\]
Note that, $\M$ verifies $\eta = \tr(\M)/\pdim = 1$. Also, when $\varrho
\downarrow 0$, then $\M$ is close to an identity matrix, and when $\varrho
\uparrow 1$, $\M$ tends to a singular matrix of rank 1. 
The dimension is fixed at $\pdim = 100$ and $\ndim$ varies from 10 to 50
in steps of 5 samples. Figure~\ref{fig:AR1_p100} depicts the NMSE
performance as a function of sample length $n$ when the samples were drawn
from a Gaussian distribution.

It can be noted that when the sample sizes were small, both the Ell1-RSCM
estimator and the Ell2-RSCM estimator outperformed the LW-RSCM estimator
with a significant margin. We also notice that the performance of the
Ell2-RSCM and Ell3-RSCM estimators were almost overlapping with the
theoretical optimal value for all values of $n$ and for both values of
$\varrho$.

Next, we consider heavier-tailed distributions than the Gaussian. Namely,
the $t_\nu$-distribution with $\nu=12$ and $\nu=7$ degrees of freedom; the
kurtosis of the marginal variable being $\mathrm{kurt}(x_i)= 0.75$ and
$\mathrm{kurt}(x_i)=2$, respectively. The results are given in
Figure~\ref{fig:AR1_p100_tsim}. 

First, we notice that Ell1-RSCM and Ell3-RSCM outperformed Ell2-RSCM and 
LW-RSCM for all values of $c=\pdim/\ndim$, $\nu$ and $\rho$.  In the
case of $\nu=7$, the performance of the Ell2-RSCM estimator declined due to
its non-robustness, and it is performing the worst among the shrinkage
estimators. In the case of $\nu=12$, the LW-RSCM estimator and the Ell2-RSCM
estimator had similar performances for larger values of $n$, but Ell2-RSCM
performed better at small values of $n$. Since Ell1-RSCM and Ell2-RSCM
differ only in the way they estimate the sphericity $\gamma$, the
performance loss of Ell2-RSCM over Ell1-RSCM can be attributed to
a larger variability and the non-robustness of the estimator $\hat
\gamma^{\text{Ell2}}$ as compared to $\hat \gamma^{\text{Ell1}}$. Also note
that when the samples were drawn from the $t_7$-distribution, the
performance loss of LW-RSCM to Ell1-RSCM and Ell3-RSCM increased.
Indeed, this difference in performance can be attributed to better
robustness properties of the Ell1-RSCM estimator over the LW-RSCM estimator
when sampling from a heavier-tailed elliptical distribution.

\begin{figure}[t]
    \centering
    \subfloat[$\varrho=0.1$ and $\nu=12$]{%
        \begin{tikzpicture}[scale=1]
            \begin{axis}[mystyle,
                xlabel = {$n$},
                ylabel = {$\mathrm{NMSE}$},
                ]
                \pgfplotstableread{AR1MSErho0p1v12.dat}\loadedtable
                \addplot[black,very thick] table [x=nlist, y=NMSEtheory]{\loadedtable};
                \addplot[blue,mark=o,thick] table [x=nlist, y=NMSEell1]{\loadedtable};
                \addplot[red,mark=triangle,thick] table [x=nlist, y=NMSEell2]{\loadedtable};
                \addplot[gray,mark=square,thick] table [x=nlist, y=NMSEell3]{\loadedtable};
                \addplot[brown,mark=star,thick] table [x=nlist, y=NMSElwe]{\loadedtable};
            \end{axis}
        \end{tikzpicture}
    }
	\hfill
    \subfloat[$\varrho=0.4$ and $\nu=12$]{%
        \begin{tikzpicture}[scale=1]
            \begin{axis}[mystyle,
                xlabel = {$n$},
                ylabel = {$\mathrm{NMSE}$},
                ]
                \pgfplotstableread{AR1MSErho0p4v12.dat}\loadedtable
                \addplot[black,very thick] table [x=nlist, y=NMSEtheory]{\loadedtable};
                \addplot[blue,mark=o,thick] table [x=nlist, y=NMSEell1]{\loadedtable};
                \addplot[red,mark=triangle,thick] table [x=nlist, y=NMSEell2]{\loadedtable};
                \addplot[gray,mark=square,thick] table [x=nlist, y=NMSEell3]{\loadedtable};
                \addplot[brown,mark=star,thick] table [x=nlist, y=NMSElwe]{\loadedtable};
            \end{axis}
        \end{tikzpicture}
    }
	\hfill
    \subfloat[$\varrho=0.1$ and $\nu=7$]{%
        \begin{tikzpicture}[scale=1]
            \begin{axis}[mystyle,
                    xlabel = {$n$},
                    ylabel = {$\mathrm{NMSE}$},
					ytick  = {0,0.2,0.4,0.6,0.8,1},
                ]
                \pgfplotstableread{AR1MSErho0p1v7.dat}\loadedtable
                \addplot[black,very thick] table [x=nlist, y=NMSEtheory]{\loadedtable};
                \addplot[blue,mark=o,thick] table [x=nlist, y=NMSEell1]{\loadedtable};
                \addplot[red,mark=triangle,thick] table [x=nlist, y=NMSEell2]{\loadedtable};
                \addplot[gray,mark=square,thick] table [x=nlist, y=NMSEell3]{\loadedtable};
                \addplot[brown,mark=star,thick] table [x=nlist, y=NMSElwe]{\loadedtable};
                \legend{Theory;,Ell1;,Ell2;,Ell3;,LW};
            \end{axis}
        \end{tikzpicture}
    }
	\hfill
    \subfloat[$\varrho=0.4$ and $\nu=7$]{%
        \begin{tikzpicture}[scale=1]
            \begin{axis}[mystyle,
                    xlabel = {$n$},
                    ylabel = {$\mathrm{NMSE}$},
                ]
                \pgfplotstableread{AR1MSErho0p4v7.dat}\loadedtable
                \addplot[black,very thick] table [x=nlist, y=NMSEtheory]{\loadedtable};
                \addplot[blue,mark=o,thick] table [x=nlist, y=NMSEell1]{\loadedtable};
                \addplot[red,mark=triangle,thick] table [x=nlist, y=NMSEell2]{\loadedtable};
                \addplot[gray,mark=square,thick] table [x=nlist, y=NMSEell3]{\loadedtable};
                \addplot[brown,mark=star,thick] table [x=nlist, y=NMSElwe]{\loadedtable};
            \end{axis}
        \end{tikzpicture}
    }
    \vspace{2em}
	\caption{AR(1) process: comparison of covariance estimators when $\pdim
		= 100$, $\varrho \in \{0.1, 0.4\}$, and the samples are drawn from a
    $t_\nu$-distribution.}\label{fig:AR1_p100_tsim}
\end{figure}
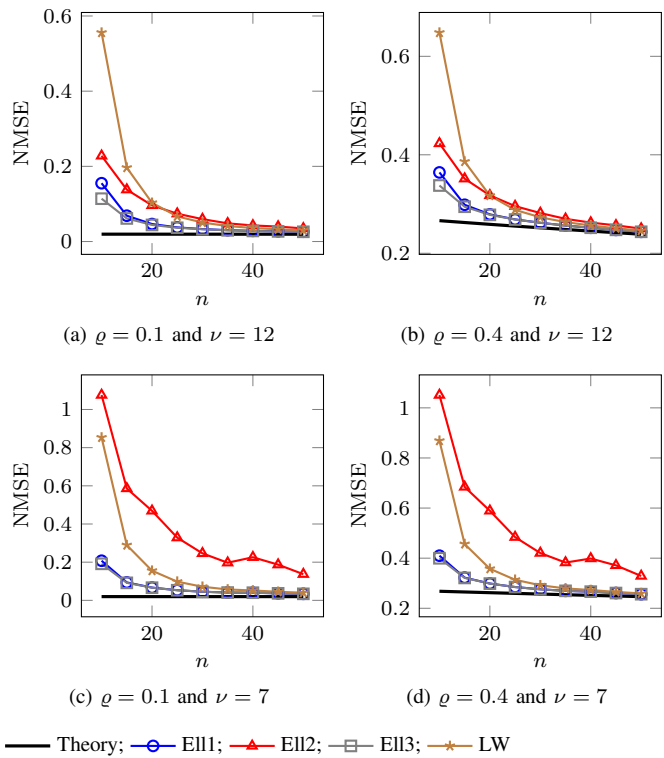

Finally, for the case $\rho=0.4$ and $n=20$, the Figure~\ref{fig:betaMSE}
depicts the theoretical NMSE of $\bo S_{(1-\be) \eta, \be}$ as a function of
$\beta$. Notice that the minimum NMSE is obtained at  $\be_o^{\textup{Ell}}$
which is shown by black vertical line. The average estimates of the optimal
value $\be_o^{\textup{Ell}}$ given by the different estimators are also
indicated by vertical lines. 

As can be seen, for Gaussian data, the Ell2-RSCM estimator of $\beta_o$ was
very close to the theoretical minimum, and significantly better than the
Ell1-RSCM estimator. The LW-RSCM estimator was far apart from the minimum
compared to Ell1-RSCM and Ell2-RSCM. In the case of the $t_\nu$-distribution
with $\nu=12$ degrees of freedom, the Ell1-RSCM estimator was performing
better than Ell2-RSCM due to its robustness, and both were significantly
closer to the minimum than the LW-RSCM estimator. 
 
\pgfplotstableread{theoreticalMSEvsBETAvInf.dat}\MSETableGaus
\pgfplotstableread{theoreticalMSEvsBETAvInfmeanBETAs.dat}\MeanBetaEllGaus
\pgfplotstablegetelem{0}{be_ora}\of{\MeanBetaEllGaus} \let\oraGaus = \pgfplotsretval
\pgfplotstablegetelem{0}{ell1}\of{\MeanBetaEllGaus}   \let\elloneGaus = \pgfplotsretval
\pgfplotstablegetelem{0}{ell2}\of{\MeanBetaEllGaus}   \let\elltwoGaus = \pgfplotsretval
\pgfplotstablegetelem{0}{ell3}\of{\MeanBetaEllGaus}   \let\ellthreeGaus = \pgfplotsretval
\pgfplotstablegetelem{0}{lwe}\of{\MeanBetaEllGaus}    \let\lweGaus = \pgfplotsretval
\pgfplotstableread{theoreticalMSEvsBETAv12.dat}\MSETable
\pgfplotstableread{theoreticalMSEvsBETAv12meanBETAs.dat}\MeanBetaEll
\pgfplotstablegetelem{0}{be_ora}\of{\MeanBetaEll} \let\ora = \pgfplotsretval
\pgfplotstablegetelem{0}{ell1}\of{\MeanBetaEll}   \let\ellone = \pgfplotsretval
\pgfplotstablegetelem{0}{ell2}\of{\MeanBetaEll}   \let\elltwo = \pgfplotsretval
\pgfplotstablegetelem{0}{ell3}\of{\MeanBetaEll}   \let\ellthree = \pgfplotsretval
\pgfplotstablegetelem{0}{lwe}\of{\MeanBetaEll}    \let\lwe = \pgfplotsretval

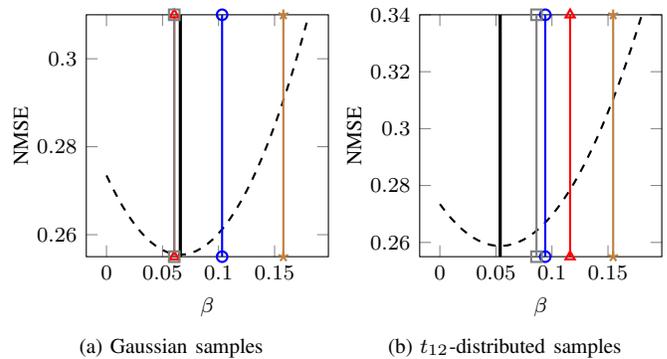
\begin{figure}
	\centering
	\subfloat[Gaussian samples]{%
		\begin{tikzpicture}[scale=1]
			\begin{axis}[
					xlabel=$\beta$,
					ylabel=NMSE,
					xticklabel style={/pgf/number format/fixed,},
					width=4.8cm, 
					height=4.8cm,
					ymin=0.255,
					ymax=0.31,
					xlabel near ticks,
					ylabel near ticks,
					tick label style={font=\footnotesize},
					label style={font=\footnotesize},
					every axis y label/.style={%
						at={(rel axis cs:-0.27,0.5)},			
						rotate=90,								
					font=\footnotesize},						
					legend style={font=\footnotesize,
						draw=none,								
						overlay,at={(-0.375,-.5)},				
					anchor=north west},							
					legend columns=-1,							
				]
				\def\ymin{0.255};
				\def\ymax{0.31};

				\addplot[black, dashed, thick] table [x=beta, y=NMSE_rscm]{\MSETableGaus};
				\addplot[black, very thick] coordinates 
					{(\oraGaus,\ymin) (\oraGaus,\ymax)};
				\addplot[blue, mark=o, thick] coordinates
					{(\elloneGaus,\ymin) (\elloneGaus,\ymax)};
				\addplot[red, mark=triangle, thick] coordinates
					{(\elltwoGaus,\ymin) (\elltwoGaus,\ymax)};
				\addplot[gray, mark=square, thick] coordinates
					{(\ellthreeGaus,\ymin) (\ellthreeGaus,\ymax)};
				\addplot[brown,mark=star, thick] coordinates 
					{(\lweGaus,\ymin) (\lweGaus,\ymax)};
				\legend{NMSE;, Theory;, Ell1;, Ell2;, Ell3;, LW};
			\end{axis}
		\end{tikzpicture}
	}
	\subfloat[$t_{12}$-distributed samples]{%
		\begin{tikzpicture}[scale=1]
			\begin{axis}[
					xlabel=$\beta$,
					ylabel=NMSE,
					xticklabel style={/pgf/number format/fixed,},
					width=4.8cm, 
					height=4.8cm,
					ymin=0.255,
					ymax=0.34,
					xlabel near ticks,
					ylabel near ticks,
					tick label style={font=\footnotesize},
					label style={font=\footnotesize},
					every axis y label/.style={%
						at={(rel axis cs:-0.27,0.5)},			
						rotate=90,								
					font=\footnotesize},						
					legend style={font=\footnotesize,
						draw=none,								
						overlay,at={(-0.3567,-.5)},				
					anchor=north west},							
					legend columns=-1,							
				]
				\def\ymin{0.255};
				\def\ymax{0.34};

				\addplot[black, dashed, thick] table [x=beta, y=NMSE_rscm]{\MSETable};
				\addplot[black, very thick] coordinates 
					{(\ora,\ymin) (\ora,\ymax)};
				\addplot[blue, mark=o, thick] coordinates
					{(\ellone,\ymin) (\ellone,\ymax)};
				\addplot[red, mark=triangle, thick] coordinates
					{(\elltwo,\ymin) (\elltwo,\ymax)};
				\addplot[gray, mark=square, thick] coordinates
					{(\ellthree,\ymin) (\ellthree,\ymax)};
				\addplot[brown,mark=star, thick] coordinates 
					{(\lwe,\ymin) (\lwe,\ymax)};
			\end{axis}
		\end{tikzpicture}
	}
	\vspace{2em}
	\caption{The theoretical NMSE of the shrinkage estimator
		$\S_{(1-\beta)\eta,\beta}$ as a function of $\beta$ when the
		covariance matrix has an AR$(1)$ structure with $\rho= 0.4$, $p =
		100$ and $n = 20$. The minimum  NMSE is obtained at
		$\be_o^{\textup{Ell}}$ which is indicated by a solid vertical line.
		The average estimated value of the shrinkage parameter $\beta$
		obtained by LW-, Ell1-, Ell2, and Ell3-estimators are shown.}
		\label{fig:betaMSE} 
\end{figure}

\subsection{Largely varying spectrum}

The next study follows the set-up in~\cite{zhang2016automatic}, where $\M$
has one (or a few) large eigenvalues. In the first set-up, $\M$ is a
diagonal matrix of size $50 \times 50$, where $m$ eigenvalues are equal to
$1$ and the remaining $50-m$ eigenvalues are equal to $0.01$. For the case
$\ndim =10$, Figure~\ref{fig:tengsim1} depicts the NMSE as a function of $m$
averaged over 10 000 Monte Carlo runs when sampling from a Gaussian
distribution and a $t_\nu$-distribution with $\nu=10$ degrees of freedom. 

In the Gaussian case, the Ell2-RSCM estimator had excellent performance as
its NMSE curve is essentially overlapping with the theoretical NMSE curve.
This is attested also in the right-hand side plot which depicts the graph of
the average estimate $\hat \be_o$ and the theoretical optimal value $\be_o$
as a function of $m$. As can be seen, the Ell2-RSCM estimator was
essentially performing at the oracle level, whereas the shrinkage parameter
corresponding to the LW-RSCM estimator was somewhat far from the theoretical
optimal. The NMSE curves show that the Ell2-RSCM estimator performed better
than the Ell1-RSCM estimator for Gaussian samples, however, with a rather
small margin. In the case of $t_{10}$-distribution, as expected, Ell1-RSCM
performed better than Ell2-RSCM due to its robustness in estimating the
sphericity. The hybrid estimator Ell3-RSCM was able to perform slightly
better than the other estimators in both cases.

\begin{figure}[t]
    \centering
    \subfloat[Gaussian samples]{%
        \begin{tikzpicture}[scale=1]
            \begin{axis}[mystyle,
                xlabel= {$m$},
                ylabel= {$\mathrm{NMSE}$},
                ]
                \pgfplotstableread{SPECTRUM1MSEvInf.dat}\loadedtable
                \addplot[black,very thick] table [x=mlist, y=NMSEtheory]{\loadedtable};
                \addplot[blue,mark=o,thick] table [x=mlist, y=NMSEell1]{\loadedtable};
                \addplot[red,mark=triangle,thick] table [x=mlist, y=NMSEell2]{\loadedtable};
                \addplot[gray,mark=square,thick] table [x=mlist, y=NMSEell3]{\loadedtable};
                \addplot[brown,mark=star,thick] table [x=mlist, y=NMSElwe]{\loadedtable};
            \end{axis}
        \end{tikzpicture}
		\hfill
        \begin{tikzpicture}[scale=1]
            \begin{axis}[mystyle,
                xlabel= {$m$},
                ylabel= {$\hat\beta$},
                ]
                \pgfplotstableread{SPECTRUM1BETAvInf.dat}\loadedtable
                \addplot[black,very thick] table [x=mlist, y=BETAtheory]{\loadedtable};
                \addplot[blue,mark=o,thick] table [x=mlist, y=meanBETAell1]{\loadedtable};
                \addplot[red,mark=triangle,thick] table [x=mlist, y=meanBETAell2]{\loadedtable};
                \addplot[gray,mark=square,thick] table [x=mlist, y=meanBETAell3]{\loadedtable};
                \addplot[brown,mark=star,thick] table [x=mlist, y=meanBETAlwe]{\loadedtable};
            \end{axis}
        \end{tikzpicture}
    }
	\hfill
    \subfloat[$t_{10}$-distributed samples]{%
        \begin{tikzpicture}[scale=1]
            \begin{axis}[mystyle,
                xlabel = {$m$},
                ylabel = {$\mathrm{NMSE}$},
				legend style={font=\footnotesize,
                  draw=none,					
                  overlay,at={(-0.3657,-.45)},	
                  anchor=north west},			
                ]
                \pgfplotstableread{SPECTRUM1MSEv10.dat}\loadedtable
                \addplot[black,very thick] table [x=mlist, y=NMSEtheory]{\loadedtable};
                \addplot[blue,mark=o,thick] table [x=mlist, y=NMSEell1]{\loadedtable};
                \addplot[red,mark=triangle,thick] table [x=mlist, y=NMSEell2]{\loadedtable};
                \addplot[gray,mark=square,thick] table [x=mlist, y=NMSEell3]{\loadedtable};
                \addplot[brown,mark=star,thick] table [x=mlist, y=NMSElwe]{\loadedtable};
                \legend{Theory;,Ell1;,Ell2;,Ell3;,LW};
            \end{axis}
        \end{tikzpicture}
		\hfill
        \begin{tikzpicture}[scale=1]
            \begin{axis}[mystyle,
                xlabel = {$m$},
                ylabel = {$\hat\beta$},
                ]
                \pgfplotstableread{SPECTRUM1BETAv10.dat}\loadedtable
                \addplot[black,very thick] table [x=mlist, y=BETAtheory]{\loadedtable};
                \addplot[blue,mark=o,thick] table [x=mlist, y=meanBETAell1]{\loadedtable};
                \addplot[red,mark=triangle,thick] table [x=mlist, y=meanBETAell2]{\loadedtable};
                \addplot[gray,mark=square,thick] table [x=mlist, y=meanBETAell3]{\loadedtable};
                \addplot[brown,mark=star,thick] table [x=mlist, y=meanBETAlwe]{\loadedtable};
            \end{axis}
        \end{tikzpicture}
    }
    \vspace{2em}
	\caption{The covariance matrix $\M$ has $m$ eigenvalues equal to $1$ and
	$50-m$ eigenvalues equal to $0.01$. Here $\pdim = 50$ and $\ndim=10$.}
    \label{fig:tengsim1}
\end{figure}
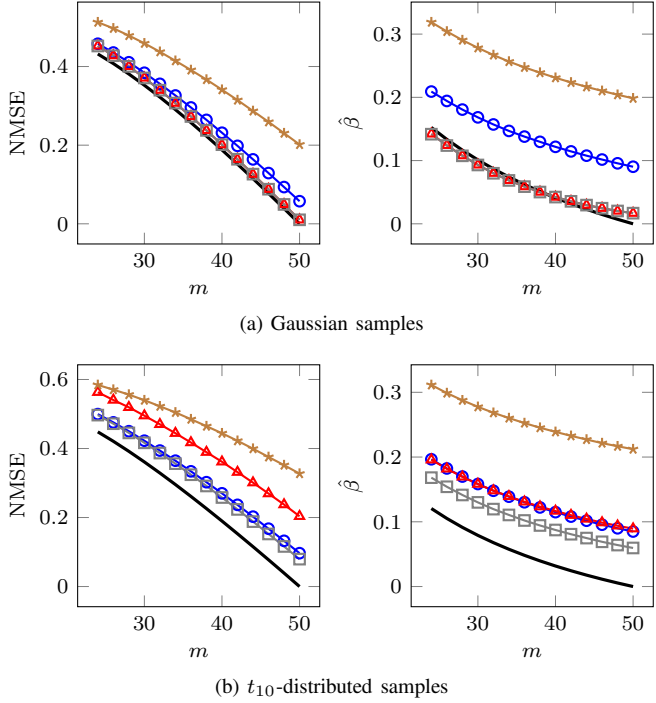

The next simulation set-up considers a very challenging scenario in which
the spectrum of $\M$ consists of several different eigenvalues. We consider
the case that $\pdim=100$ and the covariance matrix $\M$ has 30 eigenvalues
equal to $100$, $40$ eigenvalues equal to $1$, and 30 eigenvalues of $0.01$.
The samples were drawn from a Gaussian distribution and a
$t_\nu$-distribution with $\nu=10$ degrees of freedom. The NMSE curves are
plotted as a function of the sample length $\ndim$ in
Figure~\ref{fig:tengsim2}. 

It can be seen that under Gaussian sampling, the Ell2-RSCM and the Ell3-RSCM
estimators achieved near optimal performance for all $n$ considered. Indeed,
this behavior was already seen in the other simulation studies. The more
robust Ell1-RSCM estimator performed slightly worse than the Ell2-RSCM
estimator in the Gaussian case for small $n$. It can be noticed that the
performance of the LW-RSCM estimator degrades for small $n$. In the case
when the samples are from a $t_{10}$-distribution, we observe that the more
robust Ell1-RSCM estimator starts dominating the non-robust Ell2-RSCM
estimator. Again, we note that the Ell3-RSCM estimator performed the best.

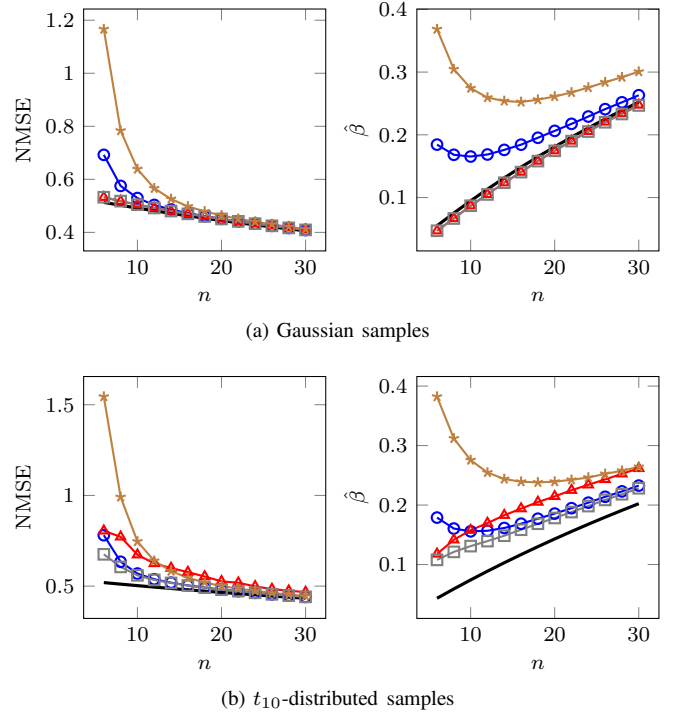
\begin{figure}[t]
    \centering
    \subfloat[Gaussian samples]{%
        \begin{tikzpicture}[scale=1]
            \begin{axis}[mystyle,
                xlabel = {$n$},
                ylabel = {$\mathrm{NMSE}$},
                ]
                \pgfplotstableread{SPECTRUM2MSEvInf.dat}\loadedtable
                \addplot[black,very thick] table [x=nlist, y=NMSEtheory]{\loadedtable};
                \addplot[blue,mark=o,thick] table [x=nlist, y=NMSEell1]{\loadedtable};
                \addplot[red,mark=triangle,thick] table [x=nlist, y=NMSEell2]{\loadedtable};
                \addplot[gray,mark=square,thick] table [x=nlist, y=NMSEell3]{\loadedtable};
                \addplot[brown,mark=star,thick] table [x=nlist, y=NMSElwe]{\loadedtable};
            \end{axis}
        \end{tikzpicture}
		\hfill
        \begin{tikzpicture}[scale=1]
            \begin{axis}[mystyle,
                xlabel = {$n$},
                ylabel = {$\hat\beta$},
                ]
                \pgfplotstableread{SPECTRUM2BETAvInf.dat}\loadedtable
                \addplot[black,very thick] table [x=nlist, y=BETAtheory]{\loadedtable};
                \addplot[blue,mark=o,thick] table [x=nlist, y=meanBETAell1]{\loadedtable};
                \addplot[red,mark=triangle,thick] table [x=nlist, y=meanBETAell2]{\loadedtable};
                \addplot[gray,mark=square,thick] table [x=nlist, y=meanBETAell3]{\loadedtable};
                \addplot[brown,mark=star,thick] table [x=nlist, y=meanBETAlwe]{\loadedtable};
            \end{axis}
        \end{tikzpicture}
    }
	\hfill
    \subfloat[$t_{10}$-distributed samples]{%
        \begin{tikzpicture}[scale=1]
            \begin{axis}[mystyle,
                xlabel = {$n$},
                ylabel = {$\mathrm{NMSE}$},
				legend style={font=\footnotesize,
                  draw=none,					
                  overlay,at={(-0.3657,-.45)},	
                  anchor=north west},			
                ]
                \pgfplotstableread{SPECTRUM2MSEv10.dat}\loadedtable
                \addplot[black,very thick] table [x=nlist, y=NMSEtheory]{\loadedtable};
                \addplot[blue,mark=o,thick] table [x=nlist, y=NMSEell1]{\loadedtable};
                \addplot[red,mark=triangle,thick] table [x=nlist, y=NMSEell2]{\loadedtable};
                \addplot[gray,mark=square,thick] table [x=nlist, y=NMSEell3]{\loadedtable};
                \addplot[brown,mark=star,thick] table [x=nlist, y=NMSElwe]{\loadedtable};
                \legend{Theory;,Ell1;,Ell2;,Ell3;,LW};
            \end{axis}
        \end{tikzpicture}
		\hfill
        \begin{tikzpicture}[scale=1]
            \begin{axis}[mystyle,
                xlabel = {$n$},
                ylabel = {$\hat\beta$},
                ]
                \pgfplotstableread{SPECTRUM2BETAv10.dat}\loadedtable
                \addplot[black,very thick] table [x=nlist, y=BETAtheory]{\loadedtable};
                \addplot[blue,mark=o,thick] table [x=nlist, y=meanBETAell1]{\loadedtable};
                \addplot[red,mark=triangle,thick] table [x=nlist, y=meanBETAell2]{\loadedtable};
                \addplot[gray,mark=square,thick] table [x=nlist, y=meanBETAell3]{\loadedtable};
                \addplot[brown,mark=star,thick] table [x=nlist, y=meanBETAlwe]{\loadedtable};
            \end{axis}
        \end{tikzpicture}
    }
    \vspace{2em}
    \caption{The covariance matrix $\M$ has 30 eigenvalues equal to $100$,
        $40$ eigenvalues equal to $1$, and 30 eigenvalues equal to $0.01$
        ($\pdim=100$).}\label{fig:tengsim2}
\end{figure}

In the last synthetic simulation study, the setup is otherwise similar to
the AR(1) setup in Subsection~\ref{sec:AR1}, but now the sample size is held
constant at $n=10$ and the degrees of freedom of the $t_\nu$-distribution of
the samples is varied from $\nu = 8$ up to $\nu = 1000$. The results are
shown in Figure~\ref{fig:tails}. One can observe that the Ell3-RSCM
estimator is able to attain the lowest empirical NMSE among all of the
estimators.

\begin{figure}[t]
	\subfloat{%
        \begin{tikzpicture}[scale=1]
            \begin{loglogaxis}[Ell3style,
                    xlabel = {$\nu$},
                    ylabel = {$\mathrm{NMSE}$},
                    xtick  = {10,15,25,50,100,500,1000},
                    ytick  = {0.3,0.4,0.5,0.6,0.7},
					every axis y label/.style={%
						at={(rel axis cs:-0.147,0.5)},
						rotate=90,
						font=\footnotesize},
                ]
                \pgfplotstableread{TAILSMSEn10p100.dat}\loadedtable
                \addplot[black,very thick] table [x=vlist, y=NMSEtheory]{\loadedtable};
                \addplot[blue,mark=o,thick] table [x=vlist, y=NMSEell1]{\loadedtable};
                \addplot[red,mark=triangle,thick] table [x=vlist, y=NMSEell2]{\loadedtable};
                \addplot[gray,mark=square,thick] table [x=vlist, y=NMSEell3]{\loadedtable};
                \addplot[brown,mark=star,thick] table [x=vlist, y=NMSElwe]{\loadedtable};
            \end{loglogaxis}
        \end{tikzpicture}
	}
	\hfill
	\subfloat{%
        \begin{tikzpicture}[scale=1]
            \begin{loglogaxis}[Ell3style,
                    xlabel = {$\nu$},
                    ylabel = {$\hat\beta$},
                    xtick  = {10,15,25,50,100,500,1000},
                    ytick  = {0,0.03,0.06,0.12,0.24},
					every axis y label/.style={%
						at={(rel axis cs:-0.1362,0.5)},
						rotate=90,
						font=\footnotesize},
                ]
                \pgfplotstableread{TAILSBETAn10p100.dat}\loadedtable
                \addplot[black,very thick] table [x=vlist, y=BETAtheory]{\loadedtable};
                \addplot[blue,mark=o,thick] table [x=vlist, y=meanBETAell1]{\loadedtable};
                \addplot[red,mark=triangle,thick] table [x=vlist, y=meanBETAell2]{\loadedtable};
                \addplot[gray,mark=square,thick] table [x=vlist, y=meanBETAell3]{\loadedtable};
                \addplot[brown,mark=star,thick] table [x=vlist, y=meanBETAlwe]{\loadedtable};
                \legend{Theory;,Ell1;,Ell2;,Ell3;,LW};
            \end{loglogaxis}
        \end{tikzpicture}
	}
    \vspace{1em}
    \caption{The NMSE and $\hat\beta$ as the $t_{\nu}$-distribution changes
    with $\nu = 8$, $10$, $12$, $15$, $25$, $50$, $100$, $500$, and $1000$.
    The plots are in log-log scale.}\label{fig:tails}
\end{figure}
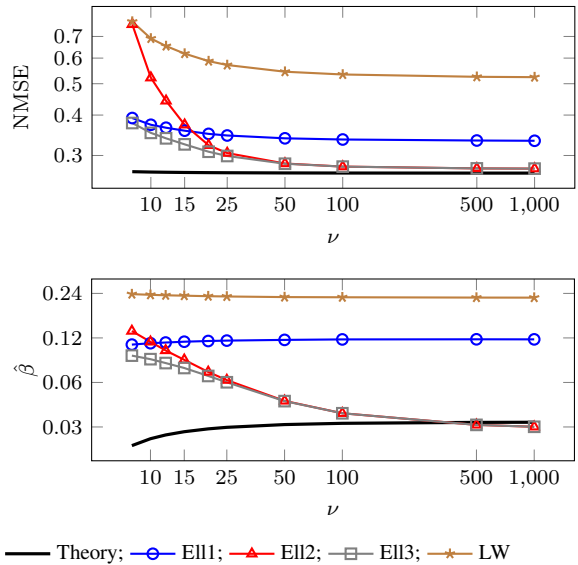

From these simulations, we can conclude that the Ell1-RSCM estimator is
better suited for heavier-tailed distributions than the Ell2-RSCM estimator,
which then again works well for Gaussian or close to Gaussian distributions.
The Ell3-RSCM estimator is, however, able to perform the best in all of the
cases. This is due to the fact that it has the freedom of choosing among two
different estimates of the sphericity; and in the conducted
simulations, the rule choosing the smaller estimate of the sphericity
turns out to work well. In the synthetic simulations, all three
proposed estimators outperformed the LW-RSCM estimator apart from the 
Ell2-RSCM estimator in the case of $t_7$-distributed samples.

\section{Data-analysis examples} \label{sec:real_data}

\subsection{Regularized QDA}

Suppose there are $\kdim$ different $\pdim$-variate populations with
covariance matrix $\M_k \in \PDHm{\pdim}$ and a mean vector $\bom \mu_k
\in\mathbb{R}^{\pdim}$, $k=1,\ldots, \kdim$. The problem is to classify an
observation $\x \in \R^\pdim$ to one of the populations.  We assume no
knowledge of the class prior probabilities. In quadratic discriminant
analysis (QDA) classification, a new observation $\x$ is assigned to class
$\hat k$ by the rule
\begin{align*}
    \hat k = \underset{k \in \{1, \ldots, \kdim\}}
    {\arg \min}{(\x-\bmu_k)}^\top {\M_k}^{-1} (\x-\bmu_k) + \log|\M_k|.
\end{align*}
Commonly, $\bmu_k$ and $\M_k$ are estimated by the sample mean vectors $\bar
\x_k$ and the SCMs $\S_k$ computed from  the \emph{training dataset}
$\X=(\x_1 \, \cdots \, \x_n)$, which consists of $n_k$ observations from
each class \mbox{$k=1,\ldots,\kdim$} and where \mbox{$n = n_1 + \cdots +
n_k$} denotes the total sample size.  In linear discriminant analysis (LDA),
one assumes that the class covariance matrices are equal, so $\M = \M_k$ for
each $k=1,\ldots, \kdim$. Then, the unknown common covariance matrix is
estimated by the \emph{pooled} SCM defined as
\[
    \S_{\mathrm{pool}}
    = \sum_{k=1}^\kdim \frac{\ndim_k-1}{\ndim - \kdim} \, \S_k.
\]
The benefit of LDA over QDA is that it can also be applied in the case when
$n_k < \pdim$ (for any $k$) as long as $n = \sum_k n_k > \pdim$. In this
case, QDA is no longer applicable since the SCM $\S_k$ is not invertible for
$n_k < \pdim$.
LDA can be viewed as a regularized form of QDA since it decreases the
variance of $\S_k$ by using the pooled SCM.\ LDA can often have superior
performance over QDA, especially in small-sample settings. 

Since the performance of LDA and QDA classification rules are highly
dependent on the covariance matrix estimates, in order to reduce the
misclassification rate, a popular approach is to use RSCM estimators
instead of the class sample covariance matrices; see
e.g.,~\cite{friedman1989regularized}. RSCMs can be applied to LDA and QDA
regardless of what the available sample sizes $n_k$ of the classes are.
Here, we use a regularized version of QDA and LDA, where we estimate the
means by $\bar \x_k$, but use Ell1-RSCM, Ell2-RSCM, or LW-RSCM in place of
the unknown covariance matrices $\M_k$ in QDA and $\M$ in case of LDA. Such
approach is referred to as regularized discriminant analysis
(RDA)~\cite{friedman1989regularized}.

We compute the misclassification rates of LDA and QDA and different RDA
methods for the phoneme dataset~\cite{ESLbook:2001}. The original data
consists of short speech frames of 32 msec duration (512 samples with at a
16kHz sampling rate) and each frame represents one of the following
phonemes, ``aa'', ``ao'',  ``dcl'', ``iy'',  or ``sh'' with the number of
occurrencies $695$, $1022$, $757$, $1163$, and $872$, respectively. The
full data set consists of $4509$ speech frames spoken by 50 different male
speakers. The data used for classification consists of the log-periodograms
of the speech frames measured at $p=256$ distinct frequencies. The goal is
to classify the spoken phonemes. 

In the simulations, we randomly split the dataset into a training set and
test set with the ratio 1:12. Then the sizes of the training sets were close
to or smaller than the dimension $\pdim$ as this is the regime where
regularization is needed the most.  The frequencies of phonemes in the
training set were $53$, $79$, $58$, $89$, and $67$, respectively, while the
remaining dataset was used as a test set. The full length of the training
data was $N= \sum_k n_k =346 > \pdim=256$, and thus, the conventional LDA
could be applied but the QDA could not be used as $n_k < p$. The
misclassification rates were calculated by classifying the observations from
the test set using the classification rule estimated from the training set.
We report the corresponding misclassification rates based on 50 repetitions
of random splits of the full data set into test sets and training sets. 
\begin{figure}[t!]
    \centering
		\begin{tikzpicture}
	\begin{axis}[myboxplotstyle]
	\addplot+[boxplot prepared={lower whisker = 0.142926,
		lower quartile = 0.160461,
		median = 0.167788,
		upper quartile = 0.179918,
		upper whisker = 0.208503}]
coordinates {};
	\addplot+[boxplot prepared={lower whisker = 0.110497,
		lower quartile = 0.121787,
		median = 0.128633,
		upper quartile = 0.139803,
		upper whisker = 0.151814}]
		table[row sep=\\,y index=0] {0.185683 \\};
	\addplot+[boxplot prepared={lower whisker = 0.112179,
		lower quartile = 0.135239,
		median = 0.143646,
		upper quartile = 0.157819,
		upper whisker = 0.177036},blue]
		table[row sep=\\,y index=0] {0.198174 \\0.207543 \\0.212347 \\};
	\addplot+[boxplot prepared={lower whisker = 0.119625,
		lower quartile = 0.142445,
		median = 0.152054,
		upper quartile = 0.173433,
		upper whisker = 0.217391},blue]
		table[row sep=\\,y index=0] {0.222436 \\};
	\addplot+[boxplot prepared={lower whisker = 0.126591,
		lower quartile = 0.173673,
		median = 0.186644,
		upper quartile = 0.206342,
		upper whisker = 0.250540},blue,solid]
		table[row sep=\\,y index=0] {0.255825 \\0.257266 \\0.265434 \\0.265914 \\};
	\addplot+[boxplot prepared={lower whisker = 0.088638,
		lower quartile = 0.095604,
		median = 0.099568,
		upper quartile = 0.105213,
		upper whisker = 0.114821},blue,solid]
		table[row sep=\\,y index=0] {0.120106 \\};
	\addplot+[boxplot prepared={lower whisker = 0.092481,
		lower quartile = 0.101369,
		median = 0.105693,
		upper quartile = 0.114581,
		upper whisker = 0.131876},blue,solid]
coordinates {};
	\addplot+[boxplot prepared={lower whisker = 0.094163,
		lower quartile = 0.101369,
		median = 0.106173,
		upper quartile = 0.114821,
		upper whisker = 0.132116},blue,solid]
coordinates {};
	\addplot+[boxplot prepared={lower whisker = 0.094883,
		lower quartile = 0.102570,
		median = 0.108576,
		upper quartile = 0.116262,
		upper whisker = 0.135719},blue,solid]
coordinates {};
	\end{axis}
\end{tikzpicture}
		\caption{Phoneme data: Box plots of the test misclassification rates
		of the conventional LDA compared with the regularized QDA and LDA
		methods based on different RSCM
		estimators.}\label{fig:boxplot_phoneme}
\end{figure}
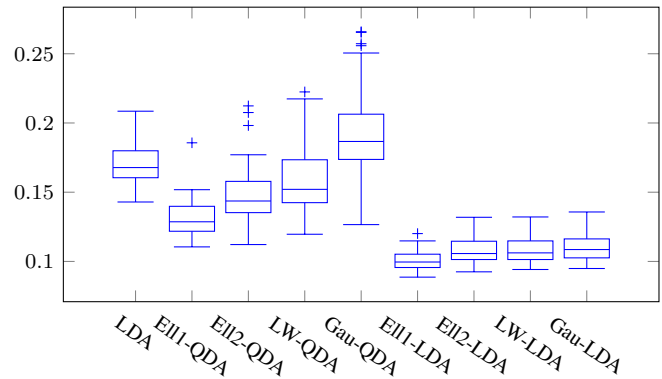
The boxplots of the test misclassification rates given in
Figure~\ref{fig:boxplot_phoneme} compare the conventional LDA with
regularized QDA and regularized LDA. Here we also compare the performance
of the Ell-RSCM estimators to an estimator that presumes Gaussianity ($\ka=0$)
and uses the shrinkage parameter estimate $\hat \be_o^{\textup{Gau}}$ as in
\eqref{eq:beta0Gau} and the estimate of the sphericity $\hat
\gam^{\textup{Ell2}}$ in place of the unknown $\gamma$. 

Several conclusions can be drawn from Figure~\ref{fig:boxplot_phoneme}.
First, the regularized LDA rules that used Ell-RSCM or LW-RSCM outperformed
the LDA with a significant margin: the median test errors of the regularized
LDA (resp. regularized QDA) methods based on Ell1-, Ell2-, and LW-RSCM
were  9.96\%,  10.57\%, and  10.62\%. (resp. 12.86\%,   14.36\%,
and 15.21\%) which may be compared with the 16.8\% median error rate of the
conventional LDA.  Second, the overall performance of the regularized LDA
methods was better than the performance of the regularized QDA methods.
Third, in all cases, both Ell1-RSCM and Ell2-RSCM outperformed LW-RSCM, and
again, Ell1-RSCM had the best performance among all methods. Fourth, we
notice that the Gau-RSCM estimator which presumes Gaussianity (and thus uses
$\ka=0$) is not able to perform better than the other RSCM estimators. In
fact, Gau-RSCM had the worst performance among all methods when applied to
the QDA rule.  This illustrates the fact that the Gaussianity assumption is
a poor approximation of reality for many real data analysis problems. 

\subsection{Portfolio optimization}
 
\begin{figure}[t]
	\centering
	\subfloat[HSI for Jan. 4, 2010 to Dec. 24, 2011.]{%
        \begin{tikzpicture}[scale=1]
		 \begin{axis}[mystyleportfolio,
                xlabel = {$\ndim$},
                ylabel = {Realized risk},
				xtick  = {50,100,150},
				ymax   = .13,
				scaled y ticks=base 10:1,
                ]
                \pgfplotstableread{minvarportfolioresultsRISK2010.dat}\loadedtable
                \addplot[blue,mark=o,thick] table [x=n, y=RiskEll1]{\loadedtable};
                \addplot[red,mark=triangle,thick] table [x=n, y=RiskEll2]{\loadedtable};
                \addplot[gray,mark=square,thick] table [x=n, y=RiskEll3]{\loadedtable};
                \addplot[brown,mark=star,thick] table [x=n, y=RiskLWE]{\loadedtable};
                \addplot[black,mark=diamond,thick] table [x=n, y=RiskROB]{\loadedtable};
              
            \end{axis}
		\end{tikzpicture}
	\hfill
        \begin{tikzpicture}[scale=1]
		 \begin{axis}[mystyleportfolio,
                xlabel = {$\ndim$},
                ylabel = {$\hat\beta$},
				xtick  = {50,100,150},
                ]
                \pgfplotstableread{minvarportfolioresultsBETA2010.dat}\loadedtable
                \addplot[blue,mark=o,thick] table [x=n, y=BETAEll1]{\loadedtable};
                \addplot[red,mark=triangle,thick] table [x=n, y=BETAEll2]{\loadedtable};
                \addplot[gray,mark=square,thick] table [x=n, y=BETAEll3]{\loadedtable};
                \addplot[brown,mark=star,thick] table [x=n, y=BETALWE]{\loadedtable};
            \end{axis}
		\end{tikzpicture}
	}
	\hfill
	\subfloat[HSI for Jan 1, 2016 to Dec. 27, 2017.]{%
        \begin{tikzpicture}[scale=1]
		 \begin{axis}[mystyleportfolio,
                xlabel = {$\ndim$},
                ylabel = {Realized risk},
				legend style={font=\footnotesize,
                  draw=none,					
                  overlay,at={(-0.375,-.45)},	
                  anchor=north west},			
                ]
                \pgfplotstableread{minvarportfolioresultsRISK2016.dat}\loadedtable
                \addplot[blue,mark=o,thick] table [x=n, y=RiskEll1]{\loadedtable};
                \addplot[red,mark=triangle,thick] table [x=n, y=RiskEll2]{\loadedtable};
                \addplot[gray,mark=square,thick] table [x=n, y=RiskEll3]{\loadedtable};
                \addplot[brown,mark=star,thick] table [x=n, y=RiskLWE]{\loadedtable};
                \addplot[black,mark=diamond,thick] table [x=n, y=RiskROB]{\loadedtable};
				\legend{Ell1;,Ell2;,Ell3;,LW;,Rob}
            \end{axis}
		\end{tikzpicture}
	\hfill
        \begin{tikzpicture}[scale=1]
		 \begin{axis}[mystyleportfolio,
                xlabel = {$\ndim$},
                ylabel = {$\hat\beta$},
				ytick = {0.8,0.9},
                ]
                \pgfplotstableread{minvarportfolioresultsBETA2016.dat}\loadedtable
                \addplot[blue,mark=o,thick] table [x=n, y=BETAEll1]{\loadedtable};
                \addplot[red,mark=triangle,thick] table [x=n, y=BETAEll2]{\loadedtable};
                \addplot[gray,mark=square,thick] table [x=n, y=BETAEll3]{\loadedtable};
                \addplot[brown,mark=star,thick] table [x=n, y=BETALWE]{\loadedtable};
            \end{axis}
		\end{tikzpicture}
	}
\vspace{2em}
\caption{Annualized realized portfolio risk and average $\hat \beta$
achieved out-of-sample for a portfolio consisting of $p=45$ stocks in HSI
for Jan. 4, 2010 to Dec. 24, 2011 (upper panel); and $p=50$ stocks for Jan.
1, 2016 to Dec. 27, 2017 (lower panel). Both time-series contain 491 trading
days.  The portfolio allocations are estimated by GMVP using different
RSCM estimators and different training window lengths $n$.  The
method of~\cite{yang2015robust} that uses a robust regularized covariance
estimator is also included and referred to as Rob.}
\label{fig:risk}
\end{figure}
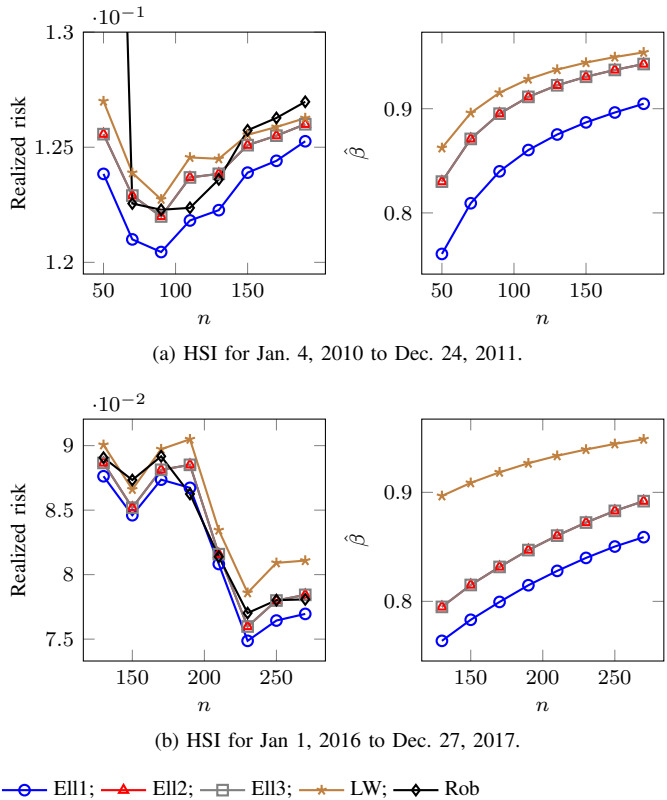

Portfolio selection and optimization is one of the most important topics in
investment theory. It is a mathematical framework wherein one seeks
portfolio allocations which balance the return-risk tradeoff such that it
satisfies the investor's needs. Some historical key references are
\cite{markowitz1952portfolio, markowitz1959portfolio, tobin1958liquidity,
sharpe1964capital}, and~\cite{lintner1965valuation}.

Consider a portfolio consisting of $\pdim$ assets. The objective is to find
optimal portfolio weights which determine the proportion of wealth that is
to be invested in each particular stock. That is, a fraction $w_i \in \R$ of
the total wealth is invested in the $i$th asset, $i=1,\ldots,\pdim$, and the
portfolio $P$ with $\pdim$ assets is described by the portfolio
\emph{weight} or \emph{allocation vector} $\w \in \R^\pdim$ which satisfies
the constraint $\bo 1^\top \w = 1$. The \emph{global mean variance
portfolio} (GMVP) aims at finding the weight vector that minimizes the
portfolio variance (risk or volatility), and hence does not require
specifying the mean vector. Let $\bo r_t \in \R^\pdim$ denote the net
returns of the $\pdim$ assets at time $t$. The GMVP optimization problem is 
\[ 
    \underset{\w \in \R^\pdim}{\mathrm{minimize}} \  \w^\top \M \w   \quad
    \mbox{subject to }  \quad \bo 1^\top \w = 1, 
\]  
where $\bo 1$  denotes is a $\pdim$-vector of ones and $\M$ denotes the
covariance matrix of the vector $\bo r_t$ of returns. The problem is
straight-forward to solve and the well-known solution is 
\begin{equation} \label{eq:w_GMVP} 
    \w_{o} = \frac{\M^{-1} \bo 1}{\bo 1^\top \M^{-1} \bo 1}.
\end{equation}
Naturally, the covariances of the net returns cannot be foreseen, and hence,
the covariance matrix needs to be estimated from the historical data. Next,
we apply the proposed RSCM estimators in the GMVP optimization application. 

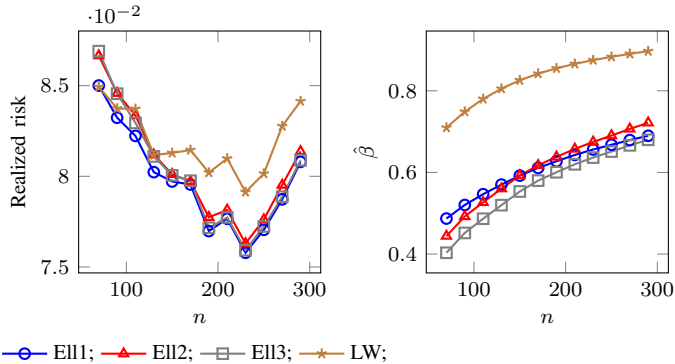
\begin{figure}[t]
	\subfloat{%
        \begin{tikzpicture}[scale=1]
		 \begin{axis}[mystyle,
                xlabel = {$\ndim$},
                ylabel = {Realized risk},
				legend style={font=\footnotesize,
                  draw=none,					
                  overlay,at={(-0.355,-.25)},	
                  anchor=north west},			
                ]
                \pgfplotstableread{SP500resultsRISK.dat}\loadedtable
                \addplot[blue,mark=o,thick] table [x=n, y=RiskEll1]{\loadedtable};
                \addplot[red,mark=triangle,thick] table [x=n, y=RiskEll2]{\loadedtable};
                \addplot[gray,mark=square,thick] table [x=n, y=RiskEll3]{\loadedtable};
                \addplot[brown,mark=star,thick] table [x=n, y=RiskLWE]{\loadedtable};
				\legend{Ell1;,Ell2;,Ell3;,LW;}
            \end{axis}
		\end{tikzpicture}
	\hfill
        \begin{tikzpicture}[scale=1]
		 \begin{axis}[mystyle,
                xlabel = {$\ndim$},
                ylabel = {$\hat\beta$},
                ]
                \pgfplotstableread{SP500resultsBETA.dat}\loadedtable
                \addplot[blue,mark=o,thick] table [x=n, y=BETAEll1]{\loadedtable};
                \addplot[red,mark=triangle,thick] table [x=n, y=BETAEll2]{\loadedtable};
                \addplot[gray,mark=square,thick] table [x=n, y=BETAEll3]{\loadedtable};
                \addplot[brown,mark=star,thick] table [x=n, y=BETALWE]{\loadedtable};
            \end{axis}
		\end{tikzpicture}
	}
\vspace{1em}
\caption{Annualized realized portfolio risk achieved out-of-sample over 583
    trading days  for a portfolio consisting of  $p=396$ stocks in S\&P 500
    index  for Jan. 4, 2016 to to Apr. 27, 2018. The portfolio allocations
    are estimated  by  GMVP using different  regularized SCM estimators and
    different training window length $n$. The right panel shows average
    $\hat \beta$ for different training windows lengths for RSCM estimators.}
    \label{fig:risk2}
\end{figure}

We investigate the out-of-sample portfolio performance of different
RSCM estimators. In particular, we use the divident adjusted
daily closing prices downloaded from the Yahoo! Finance
(\url{http://finance.yahoo.com}) database to obtain the net returns for 50
stocks that are currently included in the Hang Seng Index (HSI) for two
different time periods, from Jan. 4, 2010 to Dec. 24, 2011, and from Jan. 1,
2016 to Dec. 27, 2017  (excluding the weekends and public holidays).  In
both cases, the time series contain $T=491$ trading days. For the first
period (2010-2011), we had full length time series for only $p=45$ stocks,
whereas in the latter case we had full length time series for all stocks, so
$p=50$.  Our third time series contains the net returns of $p=396$  stocks
from Standard and Poor's 500 (SP500) index for the time period from Jan. 4,
2016 to Apr. 27, 2018 (excluding the weekends and public holidays). In this
case, the time series contains $T=583$ trading days.

At a particular day $t$, we used the previous $n$ days (i.e., from $t-n$ to
$t-1$) as the training window to estimate the covariance matrix, and the
portfolio weight vector. The obtained weight vector $\hat \w_0$ was then used
to compute the portfolio returns for the following 20 days. Next, the window
was shifted 20 trading days forward, a new weight vector was computed, and the
portfolio returns for another 20 days were computed. Hence, this scenario
corresponds to the case that the portfolio manager holds the assets for
approximately a month (20 trading days), after which they are liquidated and
new weights are computed. In this manner, we obtained $T-n$ daily returns from
which the realized risk was computed as the sample standard deviation of the
obtained portfolio returns. To obtain the annualized realized risk, the
sample standard deviations of the daily returns were multiplied by
$\sqrt{250}$. In our tests, different training window lengths $n$ were
considered. Figure~\ref{fig:risk} depicts the annualized  realized risks for the
different RSCM estimators for both time periods of the HSI data.  We also
included in our study the robust GMVP weight estimator proposed in
\cite{yang2015robust} that uses a robust regularized Tyler's $M$-estimator
with a tuning parameter selection that is optimized for the GMVP problem.
In~\cite{yang2015robust}, it was illustrated that their estimator outperforms a
large array of regularized covariance matrix estimators both for simulated
and real financial data. 

As can be seen from Figure~\ref{fig:risk}, for period 2010-2011, the
Ell1-RSCM (and Ell3-RSCM) estimator achieved the smallest realized risk,
outperforming all the other estimators for all window lengths. The robust
method of~\cite{yang2015robust} performed slightly better than the Ell2-RSCM
estimator only for certain window lengths ($n=70$ and $n=110$), but it was
also the worst method for a very small window length ($n=50$). For period
2016-2017, the differences between the estimators were not as large as in
the period 2010-2011.  Here we observed that for some window lengths, the
Ell1- and the Ell2-RSCM estimators and the robust method
of~\cite{yang2015robust} had rather identical behaviour (e.g., when
$n=210$).  Overall, however, the Ell1-RSCM method was the best performing
method. 

Finally, we wish to point out that while Ell1-RSCM was observed to have the
best performance in general, also the Ell2-RSCM estimator outperformed the
LW-RSCM over the entire span of the estimation windows considered for both
periods. Also, note that the optimal training window length which yielded
the smallest realized risk was $n=90$ for the period 2010-2011, but much
larger ($n=230$) for the period 2016-2017. This could be explained by the
fact that the stock markets were more turbulent in the first period, and
hence, the realized risks were much higher. 

Figure~\ref{fig:risk2} depicts the annualized realized risks for the
different RSCM estimators for the time period from Jan. 4, 2016 to Apr. 27,
2018 of the SP500 data. We have excluded the method of~\cite{yang2015robust}
from this study as it is not well suited for very high-dimensional problems
because of its large computational cost due to the grid search method it
uses in finding the optimal tuning parameter. 

With the SP500 data, Ell1-RSCM achieved the smallest realized risk and
outperformed the other estimators for all training window lengths $n$. The
Ell2-RSCM estimator outperformed LW-RSCM when $n \geq 130$. The Ell3-RSCM
estimator had similar performance as Ell2-RSCM when $n \leq 170$ and it
performed similar to Ell1-RSCM for $n \geq 170$. The optimal training window
length which produced the smallest realized risk was $n=230$ for all
methods.  Note that, the same result was achieved with the HSI data for the
period 2016-2017.

\section{Conclusion}\label{sec:conclusions}
This paper proposed a regularized sample covariance matrix (RSCM) estimator
Ell-RSCM, which is suitable for high-dimensional problems, where the data
can be considered as generated from an unknown elliptically symmetric
distribution. The proposed estimator is based on the estimation of the optimal
shrinkage parameters which minimizes the mean squared error. The estimation of
the optimal shrinkage parameters was shown to reduce to a simpler problem of
estimating three statistical population parameters: the scale $\eta$, the
sphericity $\gamma$, and the elliptical kurtosis $\ka$. The paper showed
alternative ways of how to estimate these parameters under the random matrix
regime. In the construction of the proposed estimator Ell-RSCM, elliptical
distribution theory was used in the derivation of the analytical form of the
mean squared error of the SCM.\ The conducted synthetic simulation studies
showed an advantage of using the proposed Ell-RSCM estimator over the widely
popular Ledoit-Wolf (LW-RSCM) estimator. 
 Furthermore, we tested the
performance of the proposed Ell-RSCM estimator using real data in a
classification problem and  in a
portfolio optimization problem, wherein the proposed methods were able to
outperform the benchmark methods. MATLAB codes of the proposed Ell-RSCM methods and codes and 
 datasets to reproduce the results of real data-analysis examples of Section~\ref{sec:real_data} are available at 
\url{http://users.spa.aalto.fi/esollila/regscm/}.

\ifCLASSOPTIONcaptionsoff
\newpage
\fi

\end{document}